\documentclass{article}

\usepackage{xspace}
\usepackage{tikz}
\usepackage{subfigure}
\usepackage{amsthm}
\usepackage{amsmath}
\usepackage{cleveref}
\usepackage{amsfonts}
\usepackage[T1]{fontenc}

\title{Fully Lattice-Linear Algorithms\footnote{\textbf{The experiments presented in this paper were supported through computational resources and services provided by the Institute for Cyber-Enabled Research, Michigan State University.
}}
\footnote{\textbf{
Excerpts of 
this paper
were published in a paper that
appeared in the 42nd International Symposium on Reliable Distributed Systems (SRDS 2023) \cite{Gupta2023}.
}}
\footnote{\textbf{This is an extended abstract of the SRDS 2023 paper. To get the arXiv version of the SRDS paper, please go to a previous version uploaded
on this arXiv repository (2210.03055v4, uploaded on 30 Jul 2023). This extended abstract is uploaded to this same arXiv repository due to the policies
of arXiv on submissions with shared content.}}
}

\author{Arya Tanmay Gupta\and Sandeep S Kulkarni}

\date{Computer Science and Engineering, Michigan State University\\
\texttt{\{atgupta,sandeep\}@msu.edu}}

\newcommand{\mds}{MDS\xspace}
\newcommand{\gc}{GC\xspace}
\newcommand{\mvc}{MVC\xspace}
\newcommand{\mis}{MIS\xspace}
\newcommand{\imped}{impedensable\xspace}
\newcommand{\Imped}{Impedensable\xspace}
\newcommand{\smp}{SMP\xspace}
\newcommand{\ds}{DS\xspace}
\newcommand{\is}{IS\xspace}
\newcommand{\vc}{VC\xspace}

\newtheorem{theorem}{Theorem}
\newtheorem{lemma}{Lemma}
\newtheorem{corollary}{Corollary}
\newtheorem{algorithm}{Algorithm}

\newtheorem{example}{Example}%

\newtheorem{examplesmpcont}{Example \smp continuation}
\newtheorem{exampledscont}{Example \mds continuation}

\newtheorem{definition}{Definition}%

\begin{document}

\maketitle

\begin{abstract}
    Lattice-linearity was introduced as a way to model problems using predicates that induce a lattice among the global states (Garg, SPAA 2020). A key property of \textit{the predicate} representing such problems is that it induces \textit{one} lattice in the state space. An algorithm that emerges from such a predicate guarantees the execution to be correct even if nodes execute asynchronously. However, many interesting problems do not exhibit lattice-linearity. This issue was somewhat alleviated with the introduction of eventually lattice-linear algorithms (Gupta and Kulkarni, SSS 2021). They induce \textit{single} or \textit{multiple} lattices in \textit{a subset of the state space} even when the problem cannot be defined by a predicate under which the global states form a lattice. 
    
    This paper focuses on analyzing and differentiating between lattice-linear problems and algorithms. We introduce \textit{fully lattice-linear algorithms}. These algorithms partition the \textit{entire} reachable state space into \textit{one or more lattices}, and as a result, ensure that the execution remains correct even if nodes execute asynchronously. For demonstration, we present lattice-linear self-stabilizing algorithms for minimal dominating set (MDS), graph colouring (GC), minimal vertex cover (MVC) and maximal independent set (MIS) problems.
    
    The algorithms for MDS, MVC and MIS converge in $n$ moves and the algorithm for GC converges in $n+2m$ moves. These algorithms preserve this time complexity while allowing the nodes to execute asynchronously. They present an improvement to the existing algorithms present in the literature.
    
    Our work also demonstrates that to allow asynchrony, a more relaxed data structure can be allowed (called $\prec$-lattice in this paper, where the meet of a pair of global states may not be defined), rather than a distributive lattice (where both join and meet for all pairs of global states are defined, and join and meet distribute over each other) as assumed by Garg. 
\end{abstract}


\section{Introduction}\label{section:introduction}

Many concurrent
algorithms can be viewed as a loop where in each step/iteration, a node reads the information about other nodes, based on which it decides to take action and update its own state. As an example, in an algorithm for graph colouring, in each step, a node reads the colour of all its neighbours and, if necessary, updates its own colour. Execution of this algorithm in a parallel or distributed system requires synchronization to ensure correct behaviour. For example, if two nodes, say $i$ and $j$ change their colour simultaneously, the resulting action may be incorrect.
Mutual exclusion, transactions, dining philosopher and locking are examples of synchronization primitives. 

Let us consider that we allow such algorithms to execute in asynchrony: consider that $i$ and $j$ from the above example execute asynchronously. Let us suppose that $i$ reads the state of $j$ and changes its state. The action of $i$ is acceptable till now. 
However, if $j$ reads the value of $i$ concurrently, it is relying on inconsistent information about $i$. This is because $i$ will have changed its state in its immediate next step. In other words, $j$ is relying on old information about $i$. 
The synchronization primitives discussed in the previous paragraph aim to eliminate such behaviour. However, they introduce considerable overhead in terms of time and computational resources.


Generally, reading values in asynchrony, as described in the example above, causes the algorithm to fail. However, if the correctness of the algorithm can be proved even when a node executes based on old information, then such an algorithm can benefit from \textit{asynchronous} execution; such execution would not suffer from synchronization overheads and each node can execute independently. 


Garg \cite{Garg2020} introduced modelling problems such that the \textit{entire reachable state space forms a single lattice}. Such induction allows asynchronous executions.
%
We \cite{Gupta2021} introduced \textit{eventually lattice-linear algorithms} for problems that cannot be modelled under the constraints of \cite{Garg2020}. These algorithms induce (potentially multiple) lattices in a \textit{subset} of the state space.
(We investigate these models comprehensively in \Cref{section:background}.)

A key property of lattice-linearity is that a total order is induced among the local states visited by individual nodes. This ensures that the local states of a node do not form a cycle; and as a consequence, single or multiple lattices are induced among the global states.

In this paper, we introduce \textit{fully} lattice-linear algorithms that are capable of inducing single or multiple lattices in the \textit{entire reachable state space}. 
We show that with fully lattice-linear algorithms, it is possible to combine lattice-linearity with self-stabilization, which ensures that the system converges to a legitimate state even if it starts from an arbitrary state.


\subsection{Contributions of the paper}

\begin{itemize}
\item We alleviate the limitations of, and bridge the gap between, \cite{Garg2020} and \cite{Gupta2021} by introducing fully lattice-linear algorithms (FLLAs). The former creates a single lattice in the state space and does not allow self-stabilization whereas the latter creates multiple lattices in a subset of the state space. FLLAs induce one or more lattices
among all the reachable global states and can enable self-stabilization. This overcomes the limitations of \cite{Garg2020} and \cite{Gupta2021}.
\item We provide upper bounds to the convergence time for an arbitrary algorithm traversing a lattice of global states.
\item We present fully lattice-linear self-stabilizing algorithms for minimal dominating set (\mds) and graph colouring (\gc).
\item We show that a direct extension of the design of the algorithms for \mds and \gc to develop an algorithm for minimal vertex cover (\mvc) exhibits cyclic behaviour. However, we exploit the properties of \mvc and make changes to the design choices to obtain an algorithm for \mvc. From here, we straightforwardly obtain an algorithm for maximal independent set (\mis) as well.
\item The algorithms for \mds, \mvc, \mis converge in $n$ moves and the algorithm for \gc converges in $n+2m$ moves. These algorithms are fully tolerant to consistency violations and asynchrony. Thus, they are an improvement over the existing algorithms in the literature. 
\end{itemize}

The major focus and benefit of this paper is to study the theory behind the algorithms, for non-lattice-linear problems, that are tolerant to asynchronous executions and induce lattices in the state space. Some applications of the specific problems studied in this paper are listed as follows.
Dominating set is applied in communication and wireless networks where it is used to compute the virtual backbone of a network.
Graph colouring is applicable in (1) chemistry, where it is used to design storage of chemicals -- a pair of reacting chemicals are not stored together, (2) communication networks, where it is used in wireless networks to compute radio frequency assignment, (3) academics, where it is used to design exam timetables -- exams attended by a single student, for each student, are not conducted at the same time, and (4) economics, where it is used in team management -- people who do not want to work together are not put in the same team. 
Vertex cover is applicable in (1) computational biology, where it is used to eliminate repetitive DNA sequences -- providing a set covering all desired sequences, and (2) economics, where it is used in camera instalments -- it provides a set of locations covering all hallways of a building.
Independent set is applied in computational biology, where it is used in discovering stable genetic components for designing engineered genetic systems.

\subsection{Organization of the paper}
In \Cref{section:preliminaries}, we discuss the preliminaries that we utilize in this paper. In \Cref{section:background}, we discuss some background results related to lattice-linearity that are present in the literature. 
In \Cref{section:ds-andcolouring-algos}, we describe the general structure of a (fully) lattice-linear algorithm. In \Cref{section:ds-ll}, we present a fully lattice-linear algorithm for minimal dominating set. We present a fully lattice-linear algorithm for graph colouring in \Cref{section:gc-ll}. 

In \Cref{section:vc-no-ll}, we discuss how a similar design of algorithms cannot be used to develop all lattice-linear algorithms.
We discuss why the design used to develop algorithms for minimal dominating set and graph colouring cannot be extended to develop algorithms for minimal vertex cover in \Cref{subsection:vc-no-tie-breaker}. We present algorithms for minimal vertex cover and maximal independent set problems in \Cref{subsection:mvc-ll} and \Cref{subsection:mis-ll} respectively. We elaborate on the properties similar in the algorithms for minimal vertex cover and maximal independent set in \Cref{subsection:complex-actions-vc-is}.

In \Cref{section:convergence-time}, we provide an upper bound to the number of moves required, for convergence, by an algorithm that traverses a lattice of states.
We discuss related work in \Cref{section:literature}. Then, in \Cref{section:experiments}, we compare the convergence time of the algorithm presented in \Cref{section:ds-ll} with other algorithms (for the minimal dominating set problem) in the literature.
Finally, we conclude in \Cref{section:conclusion}. 

\section{Preliminaries}\label{section:preliminaries}

In this paper, we are mainly interested in graph algorithms where the input is a graph $G$, $V(G)$ is the set of its nodes and $E(G)$ is the set of its edges. For a node $i\in V(G)$, $Adj_i$ is the set of nodes connected to $i$ by an edge, and $Adj^x_i$ is the set of nodes within $x$ hops from $i$, excluding $i$. While writing the time complexity of the algorithms, we notate $n$ to be $|V(G)|$ and $m$ to be $|E(G)|$. For a node $i$, $deg(i)=|Adj_i|$, and $N_i=Adj_i\cup\{i\}$.

A \textit{global state} $s$ is represented as a vector such that $s[i]$ represents the \textit{local state} of node $i$. $s[i]$ denotes the variables of node $i$, and is 
represented in the form of a vector of the variables of node $i$.
We use $S$ to denote the \textit{state space}, the set of all global states that a system can be in. 

Each node in $V(G)$ is associated with rules. Each rule at node $i$ checks the values of variables of the nodes in $Adj_i^x\cup \{i\}$ (where the value of $x$ is problem dependent) and updates the variables of $i$. A \textit{rule} at a node $i$ is of the form $g \longrightarrow a_c$ where $g$, a \textit{guard}, is a Boolean expression over variables in $Adj_i^x\cup \{i\}$ and \textit{action} $a_c$ is a set of instructions that updates the variables of $i$ if $g$ is true. 
A node is \textit{enabled} iff at least one of its guards is true, otherwise it is \textit{disabled}.
Algorithms, in this paper, are written as a sequence of rules; the nodes execute the first rule whose guard is true.
A \textit{move} is an event in which an enabled node updates its variables.

An algorithm $A$ is \textit{self-stabilizing} with respect to the subset $S_o$ of $S$ iff it satisfies the following properties: (1) \textit{convergence}: starting from an  arbitrary state, any sequence of computations of $A$ reaches a state in $S_o$, and (2) \textit{closure}: any computation of $A$ starting from $S_o$ always stays in $S_o$. 
We assume $S_o$ to be the set of \textit{optimal} states: the system is deemed converged once it reaches a state in $S_o$. $A$ is a \textit{silent} self-stabilizing algorithm if no node is enabled once a state in $S_o$ is reached.

\subsection{Execution without synchronization }

Typically, we view the \textit{computation} of an algorithm as a sequence of global states $\langle s_0, s_1, \cdots\rangle$, where $s_{t+1}, t\geq 0,$ is obtained by executing some rule by one or more nodes (as decided by the scheduler) in $s_t$.  
For the sake of discussion, assume that only node $i$ executes in state $s_t$. 
The computation prefix uptil $s_{t}$ is $\langle s_0, s_1, \cdots, s_t\rangle$. The state that the system traverses to after $s_t$ is $s_{t+1}$.
Under proper synchronization, $i$ would evaluate its guards on the \textit{current} local states of its neighbours in $s_t$, and the resultant state $s_{t+1}$ can be computed accordingly.

To understand the execution in asynchrony, let $x(s)$ be the value of some variable $x$ 
in state $s$. 
If $i$ executes in asynchrony, then it views the global state that it is in to be $s'$, 
where $x(s')\in\{ x(s_0), x(s_1), \cdots, x(s_t) \}.$
In this case, $s_{t+1}$ is evaluated as follows.
If all guards in $i$ evaluate to false, then the system will continue to remain in state $s_t$, i.e., $s_{t+1} = s_{t}$.
If a guard $g$ evaluates to true then $i$ will execute its corresponding action $a_c$.
Here, we have the following observations:
(1) $s_{t+1}[i]$ is the state that $i$ obtains after executing an action in $s'$, and (2) $\forall j\neq i$, $s_{t+1}[j] = s_t[j]$.

The model described in the above paragraph is \textit{arbitrary asynchrony}, a model in which a node can read old values of other nodes arbitrarily, requiring that if some information is sent from a node, it eventually reaches the target node.
In this paper, however, we are interested in \textit{asynchrony with monotonous read} (AMR) model. 
AMR model is arbitrary asynchrony with an additional restriction: when node $i$ reads the state of node $j$, the reads are monotonic, i.e., if $i$ reads a newer value of the state of $j$ then it cannot read an older value of $j$ at a later time. For example, if the state of $j$ changes from $0$ to $1$ to $2$ and node $i$ reads the state of $j$ to be $1$ then its subsequent read will either return $1$ or $2$, it cannot return $0$. 

\subsection{Embedding a $\prec$-lattice in global states. }

In this part, we discuss the structure of a lattice in the state space which, under proper constraints, allows an algorithm to converge in asynchrony.
To describe the embedding,
we define a total order $\prec_l$; all local states of a node $i$ are totally ordered under $\prec_l$. 
Using $\prec_l$, we define a partial order $\prec_g$ among global states as follows.

We say that $s \prec_g s^\prime$ iff $(\forall i: s[i]=s'[i]\lor s[i]\prec_l s'[i]) \land (\exists i:s[i]\prec_ls'[i])$.
Also, $s=s'$ iff $\forall i: s[i] = s'[i]$. 
For brevity, we use $\prec$ to denote $\prec_l$ and $\prec_g$: $\prec$ corresponds to $\prec_l$ while comparing local states, and $\prec$ corresponds to $\prec_g$ while comparing global states. 
We also use the symbol `$\succ$' which is such that $s\succ s'$ iff $s' \prec s$.
Similarly, we use symbols `$\preceq$' and `$\succeq$'; e.g., $s\preceq s'$ iff  $s=s' \lor s \prec s'$.
We call the lattice, formed from such partial order, a \textit{$\prec$-lattice}.

\begin{definition}\label{definition:<-lattice}
    \textit{\boldmath$\prec$-\textbf{lattice}}. 
    Given a total relation $\prec_l$ that orders the values of $s[i]$ (the local state of node $i$ in state $s$), for each $i$, the $\prec$-lattice corresponding to $\prec_l$ is defined by the following partial order:
    $s \prec s'$ iff $(\forall i: s[i] \preceq_l s'[i]) \land (\exists i: s[i] \prec_l s'[i])$.
\end{definition}

A $\prec$-lattice constraints how global states can transition among one another: a global state $s$ can transition to state $s'$ iff $s\prec s'$.

In a lattice, we can define meet and join of any two states: the meet (respectively, join), of two states $s_1$ and $s_2$ is a state $s_3$ where $\forall i, s_3[i]$ is equal to $min(s_1[i], s_2[i])$ (respectively, $max(s_1[i], s_2[i])$), where $\min(x, y) = min(y, x)=x$ iff $(x\prec_l y \lor x=y)$, and $\max(x, y) = \max(y, x)=y$ iff $(y\succ_l x \lor y=x)$.
In a $\prec$-lattice, a join can be found for any pair of global states, however, a meet may not be found for some of the pairs, the examples of which we study, in this paper, in the following sections. This makes a $\prec$-lattice an incomplete lattice.

By varying $\prec_l$ that identifies a total order among the states of a node, one can obtain different lattices. A $\prec$-lattice, embedded in the state space, is useful for permitting the algorithm to execute asynchronously.
Under proper constraints on how the lattice is formed, convergence is ensured. 

\section{Background: Types of Lattice-Linear Systems}
\label{section:background}

Lattice-linearity has been shown to be induced in problems in two ways: one, where the problem naturally manifests lattice-linearity, and another, where the problem does not manifest lattice-linearity but the algorithm imposes it. We discuss these classes of problems in \Cref{subsection:ll-problems} and  \Cref{subsection:ll-algos} respectively.

\subsection{Natural Lattice-Linearity: Lattice-Linear \textsc{Problems}}
\label{subsection:ll-problems}

In this subsection, we discuss \textit{lattice-linear problems}, i.e., the problems where the description of the problem statement creates the lattice structure. Such problems can be represented by a predicate under which the states in $S$ form a lattice.
Lattice-linear problems have been discussed in \cite{Garg2020, Garg2021, Garg2022, Gupta2023a}. 

A \textit{lattice-linear problem} $P$ can be represented by a predicate $\mathcal{P}$ such that if any node $i$ is violating $\mathcal{P}$ in a state $s$, then it must change its state. Otherwise, the system will not satisfy $\mathcal{P}$.
Let $\mathcal{P}(s)$ be true iff the global state $s$ satisfies $\mathcal{P}$. A node violating $\mathcal{P}$ in $s$ is called an \textit{\imped} node (an \textit{impediment} to progress if does not execute, \textit{indispensable} to execute for progress).

\begin{definition}\label{definition:impedensable-node}\cite{Garg2020} \textbf{\Imped node.} $\textsc{\Imped}$ $(i$, $s$, $\mathcal{P})\equiv \lnot \mathcal{P}(s)$ $\land$ $(\forall s'\succ s:s'[i]=s[i]\Rightarrow\lnot \mathcal{P}(s'))$. \end{definition}

If a node $i$ is \imped in state $s$, then in any state $s'$ such that $s'\succ s$, if the state of $i$ remains the same, then the algorithm will not converge.
Thus, predicate $\mathcal{P}$ induces a total order among the local states visited by a node, for all nodes. Consequently, the discrete structure that gets induced among the global states is a $\prec$-lattice, as described in \Cref{definition:<-lattice}. 
We say that $\mathcal{P}$, satisfying \Cref{definition:impedensable-node}, is \textit{lattice-linear} with respect to that $\prec$-lattice.

There can be multiple arbitrary lattices that can be induced among the global states. A system cannot guarantee convergence while traversing an arbitrary lattice. To guarantee convergence, we design the predicate $\mathcal{P}$ such that it fulfils some properties, and guarantees reachability to an optimal state. $\mathcal{P}$ is used by the nodes to determine if they are \imped, using \Cref{definition:impedensable-node}.
Thus, in any suboptimal global state, there will be at least one \imped node. Formally,

\begin{definition}\cite{Garg2020}\textbf{Lattice-linear predicate (LLP).}
    $\mathcal{P}$ is an LLP with respect to a $\prec$-lattice induced among the global states iff $\forall s\in S: \lnot\mathcal{P}(s) \Rightarrow \exists i:\textsc{\Imped}(i,s,\mathcal{P})$.
\end{definition}

Now we complete the definition of lattice-linear problems. In a lattice-linear problem $P$, given any suboptimal global state $s$, $P$ specifies all and the only nodes which cannot retain their local states. 
The predicate $\mathcal{P}$ is thus designed conserving this nature of the subject problem $P$.

\begin{definition}\label{definition:ll-problem}
    \textbf{Lattice-linear problems}.
    Problem $P$ is lattice-linear 
    iff there exists a predicate $\mathcal{P}$ and a $\prec$-lattice such that
    
    \begin{itemize}
        \item $P$ is deemed solved iff the system reaches a state where $\mathcal{P}$ is true,
        \item $\mathcal{P}$ is lattice-linear with respect to the $\prec$-lattice induced among the states in $S$, i.e., $\forall s: \neg \mathcal{P}(s) \Rightarrow \exists i:\textsc{\Imped}(i,s,\mathcal{P})$, and
        \item $\forall s:(\forall i:\textsc{\Imped}(i,s,\mathcal{P})\Rightarrow (\forall s':\mathcal{P}(s')\Rightarrow s'[i]\neq s[i]))$.
    \end{itemize}
\end{definition}

\noindent\textit{Remark}: A $\prec$-lattice, induced under $\mathcal{P}$, allows asynchrony because if a node, reading old values, reads the current state $s$ as $s'$, then $s'\prec s$. So $\lnot\mathcal{P}(s')\Rightarrow \lnot\mathcal{P}(s)$ because $\textsc{\Imped}(i,s',\mathcal{P})$ and $s'[i]=s[i]$.

\begin{definition}\label{definition:ssll-problem}
    \textbf{Self-stabilizing lattice-linear predicate}.
    Continuing from \Cref{definition:ll-problem},
    $\mathcal{P}$ is a self-stabilizing lattice-linear predicate if and only if the supremum of the lattice, that $\mathcal{P}$ induces, is an optimal state.
\end{definition}

\noindent Note that a self-stabilizing lattice-linear predicate $\mathcal{P}$ can also be true in states other than the supremum of the $\prec$-lattice. 

\begin{example}\label{example:mom-definition}\textbf{SMP}.
    We describe a lattice-linear problem, the \textit{stable (man-optimal) marriage problem} (\smp) from \cite{Garg2020}. In \smp, all men (respectively, women) rank women (respectively men) in terms of their preference (lower rank is preferred more). A man proposes to one woman at a time based on his preference list, and the proposal may be accepted or rejected.
    
    A global state is represented as a vector $s$ where the vector $s[i]$ contains a scalar that represents the rank of the woman, according to the preference of man $i$, whom $i$ proposes.
    
    SMP can be defined by the predicate $\mathcal{P}_{\smp}\equiv \forall m,m':m\neq m'\Rightarrow s[m]\neq s[m']$. $\mathcal{P}_{\smp}$ is true iff no two men are proposing to the same woman. A man $m$ is \imped iff there exists another man $m'$, such that $m$ and $m'$ are proposing to the same woman $w$, and $w$ prefers $m'$ over $m$. Thus, 
    $\textsc{\Imped-\smp}(m,s,\mathcal{P}_{\smp})$ $\equiv$ $\exists m'$ $:$ $s[m]=s[m']\land rank(s[m],m')< rank(s[m],m).$
    If $m$ is \imped, he increments $s[m]$ by 1 until all his choices are exhausted. Following this algorithm, an optimal state, i.e., a state where the sum of regret of men is minimized, is reached.
    \qed 
\end{example}

A key observation from the stable marriage problem (\smp) and other problems from \cite{Garg2020} is that 
the states in $S$ form \textit{one} lattice, which contains a global infimum $\ell$ and \textit{possibly} a global supremum $u$ i.e., $\ell$  and $u$ are the states such that $\forall s\in S, \ell\preceq s$ and $u \succeq s$.

\begin{examplesmpcont}
    As an illustration of \smp, consider the case where we have 3 men and 3 women. 
    The lattice induced in this case is shown in \Cref{figure:smplattice}. In this figure, every vector represents the global state $s$ such that $s[i]$ represents the rank of the woman, according to the preference of man $i$, whom $i$ proposes. 
    The algorithm begins in the state $\langle 1, 1, 1\rangle$ (i.e., each man starts with his first choice) and continues its execution in this lattice. 
    The algorithm terminates in the lowest state in the lattice where no node is \imped.
    \qed
    
\end{examplesmpcont}

\begin{figure}[ht]
    \centering 
    \begin{tikzpicture}[scale=.8,every node/.style={scale=.8}]
        \node at (0,0) (a) {$\langle$1,1,1$\rangle$};
        
        \node at (-3,1) (b) {$\langle$2,1,1$\rangle$};
        \node at (0,1) (c) {$\langle$1,2,1$\rangle$};
        \node at (3,1) (d) {$\langle$1,1,2$\rangle$};
        
        \node at (4,2) (e) {$\langle$1,1,3$\rangle$};
        \node at (2,2) (f) {$\langle$1,2,2$\rangle$};
        \node at (4,3) (g) {$\langle$1,2,3$\rangle$};
        
        \node at (0,2) (h) {$\langle$1,3,1$\rangle$};
        \node at (2,3) (i) {$\langle$1,3,2$\rangle$};
        \node at (4,4) (j) {$\langle$1,3,3$\rangle$};
        
        \node at (-4,2) (k) {$\langle$2,1,2$\rangle$};
        \node at (-2,4) (l) {$\langle$2,1,3$\rangle$};
        
        \node at (-2,2) (m) {$\langle$2,2,1$\rangle$};
        \node at (0,3) (n) {$\langle$2,2,2$\rangle$};
        \node at (2,5) (o) {$\langle$2,2,3$\rangle$};
        
        \node at (-2,3) (p) {$\langle$2,3,1$\rangle$};
        \node at (0,5) (q) {$\langle$2,3,2$\rangle$};
        \node at (3,6) (r) {$\langle$2,3,3$\rangle$};
        
        \node at (-4,3) (s) {$\langle$3,1,1$\rangle$};
        \node at (2,4) (t) {$\langle$3,1,2$\rangle$};
        \node at (-2,5) (u) {$\langle$3,1,3$\rangle$};
        
        \node at (-4,4) (v) {$\langle$3,2,1$\rangle$};
        \node at (4,5) (w) {$\langle$3,2,2$\rangle$};
        \node at (0,6) (x) {$\langle$3,2,3$\rangle$};
        
        \node at (-4,5) (y) {$\langle$3,3,1$\rangle$};
        \node at (-3,6) (z) {$\langle$3,3,2$\rangle$};
        \node at (0,7) (parent) {$\langle$3,3,3$\rangle$};
        
        \draw (a) -- (b);
        \draw (a) -- (c);
        \draw (a) -- (d);
        \draw (d) -- (e);
        \draw (d) -- (f);
        \draw (c) -- (f);
        \draw (e) -- (g);
        \draw (f) -- (g);
        \draw (c) -- (h);
        \draw (f) -- (i);
        \draw (h) -- (i);
        \draw (g) -- (j);
        \draw (i) -- (j);
        \draw (b) -- (k);
        \draw (d) -- (k);
        \draw (e) -- (l);
        \draw (k) -- (l);
        \draw (b) -- (m);
        \draw (c) -- (m);
        \draw (m) -- (n);
        \draw (n) -- (o);
        \draw (g) -- (o);
        \draw (l) -- (o);
        \draw (f) -- (n);
        \draw (k) -- (n);
        \draw (h) -- (p);
        \draw (m) -- (p);
        \draw (p) -- (q);
        \draw (i) -- (q);
        \draw (n) -- (q);
        \draw (o) -- (r);
        \draw (j) -- (r);
        \draw (q) -- (r);
        \draw (b) -- (s);
        \draw (k) -- (t);
        \draw (s) -- (t);
        \draw (t) -- (u);
        \draw (l) -- (u);
        \draw (s) -- (v);
        \draw (m) -- (v);
        \draw (n) -- (w);
        \draw (t) -- (w);
        \draw (v) -- (w);
        \draw (w) -- (x);
        \draw (u) -- (x);
        \draw (o) -- (x);
        \draw (p) -- (y);
        \draw (v) -- (y);
        \draw (y) -- (z);
        \draw (w) -- (z);
        \draw (q) -- (z);
        \draw (z) -- (parent);
        \draw (r) -- (parent);
        \draw (x) -- (parent);
    \end{tikzpicture}
    \caption{Lattice for \smp with 3 men and 3 women; $\ell=(1,1,1)$ and $u=(3,3,3)$. Transitive edges are not shown for brevity.}
    \label{figure:smplattice}
\end{figure}

In \smp and other problems in \cite{Garg2020}, the algorithm needs to be initialized to $\ell$ to reach an optimal solution.
If we start from a state $s$ such that $s \neq \ell$, then the algorithm can only traverse the lattice from $s$. Hence, upon termination, it is possible that the optimal solution is not reached. 
In other words, such algorithms cannot be self-stabilizing unless $u$ is optimal. 

\begin{examplesmpcont}
    Consider that men and women are $M=(A,J,T)$ and $W=(K,Z,M)$. Let that proposal preferences of men are $A=(Z,K,M)$, $J=(Z,K,M)$ and $T=(K,M,Z)$, and women have ranked men as $Z=(A,J,T)$, $K=(J,T,A)$ and $M=(T,J,A)$. The optimal state (starting from $\langle 1,1,1\rangle$) is $\langle 1,2,2\rangle$.
    Starting from $\langle 1,2,3\rangle$, the algorithm terminates at $\langle 1,2,3\rangle$ which is not optimal.
    Starting from $\langle 3,1,2\rangle$, the algorithm terminates declaring that no solution is available.
    \qed 
\end{examplesmpcont}

\subsection{Imposed Lattice-Linearity: Eventually Lattice-Linear \textsc{Algorithms}}\label{subsection:ll-algos}

Unlike the lattice-linear problems where the problem description creates a lattice among the states in $S$, there are problems where the states do not form a lattice naturally, i.e., in those problems, given a suboptimal global state, the problem does not specify a specific set of nodes to change their state. As a result, in such problems, there are instances in which the \imped nodes cannot be determined naturally, i.e., in those instances
$\exists s :\lnot\mathcal{P}(s) \wedge   (\forall i : \exists s' : \mathcal{P}(s')\land s[i]=s'[i]$). 

However, lattices can be induced in the state space algorithmically in these cases. In \cite{Gupta2021}, we presented algorithms for some of such problems.
Specifically, the algorithms presented in \cite{Gupta2021} partition the state space into two parts: feasible and infeasible states, and induce multiple lattices among the feasible states. These algorithms work in two phases. The first phase takes the system from an infeasible state to a feasible state (where the system starts to exhibit the desired property), which is an element of a lattice. In the second phase, only an \imped node can change its state. This phase takes the system from a feasible state to an optimal state. These algorithms converge starting from an arbitrary state; they are called \textit{eventually lattice-linear self-stabilizing algorithms.}

\begin{example}\textbf{\mds}.\label{example:dominating-set-definition}
    In the minimal dominating set problem, the task is to choose a minimal set of nodes $\mathcal{D}$ in a given graph $G$ such that for every node in $V(G)$, either it is in $\mathcal{D}$, or at least one of its neighbours is in $\mathcal{D}$. Each node $i$ stores a variable $i[st]$ with domain $\{IN, OUT\}$; $i\in\mathcal{D}$ iff $i[st]=IN$.
    \qed
\end{example}
    
\noindent\textit{Remark}: The minimal dominating set (MDS) problem is not a lattice-linear problem. This is because, for any given node $i$, an optimal state can be reached if $i$ does or does not change its state. Thus $i$ cannot be deemed as \imped or not \imped under the natural constraints of \mds.
    
\begin{exampledscont}
    Even though \mds is not a lattice-linear problem, lattice-linearity can be imposed on it algorithmically. \Cref{algorithm:ds-ellss} (present in the following) is based on the algorithm in \cite{Gupta2021} for a more generalized version of the problem, the service demand based minimal dominating set problem. \Cref{algorithm:ds-ellss} consists of two phases. In the first phase, if node $i$ is \textit{addable}, i.e., if $i$ and all its neighbours are not in dominating set (\ds) $\mathcal{D}$, then $i$ enters $\mathcal{D}$. 
    This phase does not satisfy the constraints of lattice-linearity from \Cref{definition:impedensable-node}.
    However, once the algorithm reaches a state where nodes in $\mathcal{D}$ form a (possibly non-minimal) \ds, phase 2 imposes lattice-linearity. Specifically, in phase 2, a node $i$ leaves $\mathcal{D}$ iff it is \textit{\imped}, i.e., $i$ along with all neighbours of $i$ stay dominated even if $i$ moves out of $\mathcal{D}$, and $i$ is of the highest ID among all the removable nodes within its distance-2 neighbourhood. 
    \qed 
    
    \begin{algorithm}\label{algorithm:ds-ellss} Eventually lattice-linear algorithm for \mds.
        \begin{center}
            \begin{tabular}{|l|}
                \hline
                \textsc{Addable-DS}$(i)$ $\equiv i[st]=OUT\land$\\ \quad\quad\quad\quad $(\forall j\in Adj_i:j[st]=OUT)$.\\
                \textsc{Removable-DS}$(i)$ $\equiv i[st]=IN\land (\forall j\in Adj_i\cup\{i\}:$\\ \quad\quad\quad\quad $((j\neq i\land j[st]=IN)\lor$\\\quad\quad\quad\quad $(\exists k\in Adj_j, k\neq i: k[st]=IN)))$.\\
                \scriptsize{
                // Node $i$ can be removed without violating dominating set
                }\\
                \textsc{\Imped-DS}$(i) \equiv \textsc{Removable-DS}(i) \land$\\ \quad\quad\quad\quad $(\forall j\in Adj^2_i:\lnot \textsc{Removable-DS}(j) \lor$\\ \quad\quad\quad\quad $i[id]>j[id])$.\\~\\
                Rules for node $i$:\\
                \textsc{Addable-DS}$(i)$ $\longrightarrow i[st]=IN$. (phase 1)\\
                \textsc{\Imped-DS}$(i)$ $\longrightarrow i[st] = OUT$. (phase 2)\\
                \hline
            \end{tabular}
        \end{center}
    \end{algorithm}
\end{exampledscont}

\begin{exampledscont}\label{example:4-nodes}
    To illustrate the lattice imposed by phase 2 of \Cref{algorithm:ds-ellss}, consider an example graph $G_4$ with four nodes such that they form two disjoint edges, i.e., $V(G)=\{v_1,v_2,v_3,v_4\}$ and $E(G)=\{\{v_1,v_2\},\{v_3,v_4\}\}$. Assume that $G$ is initialized in a feasible state.
    
    The lattices formed in this case are shown in \Cref{figure:half-lattices-from-ds-example}. We write a state $s$ of this graph as $\langle v_1[st], v_2[st], v_3[st], v_4[st]\rangle$. We assume that $v_i[id]>v_j[id]$ iff $i>j$.
    Due to phase 1, the nodes not being dominated move in the DS, which makes the system traverse to a feasible state. Now, due to phase 2, only \imped nodes move out, thus, lattices are induced among the feasible global states as shown in the figure.
    \qed
\end{exampledscont}
\begin{figure}[ht]
    \centering
    \subfigure[]{
        \begin{tikzpicture}[scale=.7,every node/.style={scale=.7}]
            \node at (0,0) (a1) {$\langle$IN,OUT,IN,OUT$\rangle$};
            \node at (-1.5,-1) (a2) {$\langle$IN,OUT,IN,IN$\rangle$};
            \node at (1.5,-1) (a3) {$\langle$IN,IN,IN,OUT$\rangle$};
            \node at (0,-2) (a4) {$\langle$IN,IN,IN,IN$\rangle$};
            \draw (a1) -- (a2);
            \draw (a1) -- (a3);
            \draw (a2) -- (a4);
            \draw (a3) -- (a4);
        \end{tikzpicture}
    }\quad\quad 
    \subfigure[]{
        \begin{tikzpicture}[scale=.7,every node/.style={scale=.7}]
            \node at (0,0) (a1) {$\langle$OUT,IN,OUT,IN$\rangle$};
            \node at (0,-1) (a2) {//only 1 state};
        \end{tikzpicture}
    }\\
    \subfigure[]{
        \begin{tikzpicture}[scale=.7,every node/.style={scale=.7}]
            \node at (0,0) (a1) {$\langle$OUT,IN,IN,OUT$\rangle$};
            \node at (0,-1) (a2) {$\langle$OUT,IN,IN,IN$\rangle$};
            \draw (a1) -- (a2);
        \end{tikzpicture}
    }\quad\quad 
    \subfigure[]{
        \begin{tikzpicture}[scale=.7,every node/.style={scale=.7}]
            \node at (0,0) (a1) {$\langle$IN,OUT,OUT,IN$\rangle$};
            \node at (0,-1) (a2) {$\langle$IN,IN,OUT,IN$\rangle$};
            \draw (a1) -- (a2);
        \end{tikzpicture}
    }
    \caption{The lattices induced in the problem instance in Example MDS continuation \ref{example:4-nodes}. Transitive edges are not shown for brevity.}
    \label{figure:half-lattices-from-ds-example}
\end{figure}

In \Cref{algorithm:ds-ellss}, lattices are induced among only some of the global states. After the execution of phase 1, the algorithm locks into one of these lattices.
Thereafter in phase 2, the algorithm executes lattice-linearly to reach the supremum of that lattice. Since the supremum of every lattice represents an \mds, this algorithm always converges to an optimal state. 

\section{Introducing Fully Lattice-Linear Algorithms: Overcoming Limitations of \cite{Garg2020} and \cite{Gupta2021}}

\label{section:ds-andcolouring-algos}

In this section, we introduce fully lattice-linear algorithms that induce a lattice structure among all reachable states. While defining these algorithms, we also distinguish them from the closely related work in \cite{Garg2020} and \cite{Gupta2021}. 
Specifically, we discuss why developing algorithms for non-lattice-linear problems (such that the algorithms are lattice-linear, i.e., they induce lattices in the reachable state space) requires the innovation presented in this paper.

In \cite{Garg2020}, authors consider lattice-linear problems. Here, the state space is induced under a predicate and forms one lattice. Such problems, for a given suboptimal state, specify a set of nodes that must change their state, in order for the system to reach an optimal state. Such problems possess only one optimal state, and hence a violating node must change its state. The acting algorithm simply follows that lattice to reach the optimal state. 

Certain problems, e.g., dominating set are not lattice-linear 
(cf. the remark below \Cref{example:dominating-set-definition}) and thus they cannot be modelled under the constraints of \cite{Garg2020}. That is, the problem cannot specify for an arbitrary suboptimal state, a specific set of nodes that must change their state.
Such problems are studied in \cite{Gupta2021}. An interesting observation on the algorithms studied in \cite{Gupta2021} is that they induce multiple lattices in a subset of the state space (c.f. \Cref{figure:half-lattices-from-ds-example}). 

\textit{Limitations of \cite{Garg2020}:}
From the above discussion, we note that the general approach presented in \cite{Gupta2021} is applicable to a wider class of problems. Additionally, many lattice-linear problems do not allow self-stabilization. In such cases, e.g., in SMP, if the algorithm starts in, e.g., the supremum of the lattice, then it may terminate declaring that no solution is available. Unless the supremum is the optimal state, the acting algorithm cannot be self-stabilizing.

\textit{Limitations of \cite{Gupta2021}:}
In eventually lattice-linear algorithms (e.g., \Cref{algorithm:ds-ellss} for \mds), the lattice structure is imposed only on a subset of states. Thus, by design, the algorithm has a set of rules, say $\mathcal{A}_1$, that operate in the part of the state space where the lattice structure does not exist, and another set of rules, say $\mathcal{A}_2$, that operate in the part of the state space where the lattice is induced. 
Since actions of $\mathcal{A}_1$ operate outside the lattice structure, a developer must guarantee that if the system is initialized outside the lattice structure, then $\mathcal{A}_1$ converges the system to one of the states participating in the lattice (from where $\mathcal{A}_2$ will be responsible for the traversal of the system through the lattice) and thus the developer faces an extra proof obligation. 
In addition, it also must be proven that 
actions of $\mathcal{A}_2$ and the actions of $\mathcal{A}_1$ do not interfere with each other. E.g., the developer (in the context of \Cref{algorithm:ds-ellss}) has to make sure that actions of $\mathcal{A}_2$ do not perturb the node to a state where the selected nodes do not form a dominating set.

\textit{Alleviating the limitations of \cite{Garg2020} and \cite{Gupta2021}:} 
In this paper, we investigate if we can benefit from the advantages of both \cite{Garg2020} and \cite{Gupta2021}. We study if there exist fully lattice-linear algorithms where lattices can be imposed on all reachable states, forming single or multiple lattices. In the case that there are multiple optimal states and the problem requires self-stabilization, it would be necessary that multiple disjoint lattices are formed where the supremum of each lattice is an optimal state.
Self-stabilization also requires that these lattices are exhaustive, i.e., they collectively contain all states in the state space. 

Incorporating the property of self-stabilization ensures that the system can be allowed to initialize in any state, and asynchrony can be permitted. The initial state locks into one of the lattices, and due to the induction of $\prec$-lattices, such algorithms ensure a deterministic output (all local states visited by individual nodes form a total order, so an \imped node has only one choice of action, and thus, the global state of convergence can be predicted deterministically from the initial state or any intermediate state; the following sections contain examples of such algorithms).
Such algorithms would also permit multiple optimal states.
In addition, there will be no need to deal with interference between actions. 

\begin{definition}\label{definition:ll-algos}\textbf{Lattice-linear algorithms (LLA)}.
    Algorithm $A$ is an LLA for a problem $P$, iff there exists a predicate $\mathcal{P}$ and $A$ induces a $\prec$-lattice among the states of $S_1, ..., S_w \subseteq S (w\geq 1)$, such that
    \begin{itemize}
        \item State space $S$ of $P$ contains mutually disjoint lattices, i.e.
        \begin{itemize}
            \item $S_1, S_2, \cdots, S_w\subseteq S$ are pairwise disjoint.
            \item $S_1 \cup \cdots \cup S_w$ contains all the reachable states (starting from a set of initial states, if specified; if an arbitrary state can be an initial state, then $S_1 \cup \cdots \cup S_w=S$).
        \end{itemize}
        \item Lattice-linearity is satisfied in each subset under $\mathcal{P}$, i.e., 
        \begin{itemize}
            \item $P$ is deemed solved iff the system reaches a state where $\mathcal{P}$ is true
            \item $\forall k: 1 \leq k \leq w$, 
            $\mathcal{P}$ is lattice-linear with respect to the $\prec$-lattice induced in $S_k$ by $A$, i.e., $\forall s\in S_k: \lnot\mathcal{P}(s) \Rightarrow \exists i:
            \textsc{\Imped}(i,s,\mathcal{P})$.
        \end{itemize}
    \end{itemize}
\end{definition}

\noindent\textit{Remark}: Any algorithm that traverses a $\prec$-lattice of global states is a lattice-linear algorithm. An algorithm that solves a lattice-linear problem, under the constraints of lattice-linearity, e.g. the algorithm described in \Cref{example:mom-definition}, is also a lattice-linear algorithm.

\begin{definition}\textit{Self-stabilizing LLA}.
    Continuing from \Cref{definition:ll-algos}, $A$ is self-stabilizing only if $S_1 \cup S_2 \cup \cdots \cup S_w=S$ and $\forall k:1\leq k\leq w$, the supremum of the lattice induced among the states in $S_k$ is optimal.
\end{definition}

\section{Fully Lattice-Linear Algorithm for Minimal Dominating Set (\mds)}\label{section:ds-ll}

In this section, we present a lattice-linear self-stabilizing algorithm for \mds. \mds has been defined in \Cref{example:dominating-set-definition}.

We describe the algorithm as \Cref{algorithm:ds-ll}.
The first two macros are the same as \Cref{example:dominating-set-definition}. The definition of a node being \imped is changed to make the algorithm fully lattice-linear. Specifically, even allowing a node to enter into the dominating set (DS) is restricted such that only the nodes with the highest ID in their distance-2 neighbourhood can enter the \ds. 
Any node $i$ which is addable or removable will toggle its state iff it is \imped, i.e., iff any other node $j\in Adj^2_i:j[id]>i[id]$ is neither addable nor removable. In the case that $i$ is \imped, if $i$ is addable, then we call it \textit{addable-\imped}, otherwise, if it is removable, then we call it \textit{removable-\imped}.

\begin{algorithm}\label{algorithm:ds-ll}Algorithm for \mds.
    \begin{center}
    \begin{tabular}{|l|}
        \hline 
        \textsc{Removable-DS}$(i)$ $\equiv i[st]=IN\land (\forall j\in Adj_i\cup\{i\}:$\\ \quad\quad $((j\neq i\land j[st]=IN)\lor$\\\quad\quad $(\exists k\in Adj_j, k\neq i: k[st]=IN)))$.\\
        \textsc{Addable-DS}$(i)$ $\equiv i[st]=OUT\land$\\ \quad\quad $(\forall j\in Adj_i:j[st]=OUT)$.\\
        \textsc{Unsatisfied-DS}$(i)$ $\equiv \textsc{Removable-DS}(i)\lor$\\ \quad\quad $\textsc{Addable-DS}(i)$.\\
        $\textsc{\Imped-II-DS}(i)\equiv\textsc{Unsatisfied-DS}(i)\land$\\ \quad\quad $(\forall j\in Adj^2_i:\lnot\textsc{Unsatisfied-DS}(j)\lor$ $ i[id]>j[id])$.\\~\\
        Rules for node $i$.\\
        $\textsc{\Imped-II-DS}(i) \longrightarrow i[st] = \lnot i[st]$.\\
        \hline 
    \end{tabular}
    \end{center}
\end{algorithm}

\begin{lemma}\label{lemma:ds-no-step-back}
    Any node in an input graph does not revisit its older state while executing under \Cref{algorithm:ds-ll}.
\end{lemma}

\begin{proof}
    Let $s$ be the global state at time $t$ while \Cref{algorithm:ds-ll} is executing. We have from \Cref{algorithm:ds-ll} that if a node $i$ is addable-\imped or removable-\imped, then no other node in $Adj^2_i$ changes its state.
    
    If $i$ is addable-\imped at $t$, then any node in $Adj_i$ is out of the \ds. After when $i$ moves in, then any other node in $Adj_i$ is no longer addable, so they do not move in after $t$. As a result, $i$ does not have to move out after moving in.
    
    If otherwise $i$ is removable-\imped at $t$, then all the nodes in $Adj_i\cup\{i\}$ are being dominated by some node other than $i$. So after when $i$ moves out, then none of the nodes in $Adj_i$, including $i$, becomes unsatisfied. 
    
    Let that $i$ is dominated and out, and some $j\in Adj_i$ is removable \imped. $j$ will change its state to $OUT$ only if $i$ is being covered by another node. Also, while $j$ turns out of the \ds, no other node in $Adj^2_j$, and consequently in $Adj_i$, changes its state. As a result, $i$ does not have to turn itself in because of the action of $j$.
    
    From the above cases, we have that $i$ does not change its state to $i[st]$ after changing its state from $i[st]$ to $i[st']$.
    throughout the execution of \Cref{algorithm:ds-ll}.
\end{proof}

To demonstrate that \Cref{algorithm:ds-ll} is lattice-linear, we define state value and rank, the auxiliary variables associated with nodes and global states, as follows:
$$
    \begin{array}{l}
        \textsc{State-Value-DS}(i,s)=\\
        \begin{cases}
            1 & \text{if $\textsc{Unsatisfied-DS}(i)$ in state $s$} \\
            0 & \text{otherwise}
        \end{cases}
    \end{array}
$$
$$
    \textsc{Rank-DS}(s)=\sum\limits_{i\in V(G)}\textsc{State-Value-DS}(i,s).
$$

\begin{theorem}\label{theorem:ds-ll}
    \Cref{algorithm:ds-ll} is a silent self-stabilizing and lattice-linear algorithm executed by $n$ nodes running asynchronously.
\end{theorem}

\begin{proof}
    We have from the proof of \Cref{lemma:ds-no-step-back} that if $G$ is in state 
    $s$ and $\textsc{Rank-DS}(s)$ is non-zero, then at least one node will be \imped, e.g., the unsatisfied node in $V(G)$ with the highest ID.
    For any node $i$, we have that $\textsc{State-Value-DS}(i)$ decreases whenever $i$ is \imped and never increases.
    As a result, $\textsc{Rank-DS}$ monotonously decreases throughout the execution of the algorithm until it becomes zero. This shows that \Cref{algorithm:ds-ll} is self-stabilizing. Once \textsc{Rank-DS} is zero, no node is \imped, so no node makes a move. This shows that \Cref{algorithm:ds-ll} is silent.
    
    Next, we show that \Cref{algorithm:ds-ll} is fully lattice-linear. We claim that there is one lattice corresponding to each optimal state. It follows that if there are $w$ optimal states for a given instance, then there are $w$ disjoint lattices $S_1, S_2, \cdots, S_w$ formed in the state space $S$. We show this as follows.
    
    We observe that an optimal state (manifesting a minimal dominating set) is at the supremum of its respective lattice, as there are no outgoing transitions from an optimal state. 
    
    Furthermore, given a state $s$, we can uniquely determine the optimal state that would be reached from $s$.  
    This is because in any given non-optimal state, the \imped nodes, that must change their state in order to reduce the ranks of the global state of the system, can be uniquely identified.  Additionally, the value, that these \imped nodes will update their local state to, is also unique. Thus, the optimal state reached from a given state $s$ can be uniquely identified. 
    
    This implies that starting from a state $s$ in $S_k (1\leq k\leq w)$, the algorithm cannot converge to any state other than the supremum of $S_k$. Thus, the state space of the problem is partitioned into $S_1, S_2\cdots, S_w$ where each subset $S_k$ contains one optimal state, say $s_{k_{opt}}$, and from all states in $S_k$, the algorithm converges to $s_{k_{opt}}$. 
    
    Each subset, $S_1, S_2, \cdots, S_w$, forms a $\prec$-lattice where $s[i]\prec s'[i]$ iff $\textsc{State-Value-DS}(i,s)>\textsc{State-Value-DS}(i,s')$ and $s\prec s'$ iff $\textsc{Rank-DS}(s)>\textsc{Rank-DS}(s')$. 
    This 
    shows that \Cref{algorithm:ds-ll} is lattice-linear.
\end{proof}

\begin{exampledscont}\label{example:fullylatticelinear}
    For $G_4$ the lattices induced under \Cref{algorithm:ds-ll} are shown in \Cref{figure:full-lattices-from-ds-example}; each vector represents a global state $\langle v_1[st]$, $v_2[st]$, $v_3[st]$, $v_4[st]\rangle$. 
    \qed 
\end{exampledscont}

\begin{figure}[ht]
    \centering
    \subfigure[]{
        \begin{tikzpicture}[scale=.7,every node/.style={scale=.6}]
            \node at (0,0) (a1) {$\langle$IN,OUT,IN,OUT$\rangle$};
            \node at (-1.5,-1) (a2) {$\langle$IN,OUT,IN,IN$\rangle$};
            \node at (1.5,-1) (a3) {$\langle$IN,IN,IN,OUT$\rangle$};
            \node at (0,-2) (a4) {$\langle$IN,IN,IN,IN$\rangle$};
            \draw (a1) -- (a2);
            \draw (a1) -- (a3);
            \draw (a2) -- (a4);
            \draw (a3) -- (a4);
        \end{tikzpicture}
    }
    \subfigure[]{
        \begin{tikzpicture}[scale=.7,every node/.style={scale=.6}]
            \node at (0,0) (a1) {$\langle$OUT,IN,OUT,IN$\rangle$};
            \node at (-1.5,-1) (a2) {$\langle$OUT,IN,OUT,OUT$\rangle$};
            \node at (1.5,-1) (a3) {$\langle$OUT,OUT,OUT,IN$\rangle$};
            \node at (0,-2) (a4) {$\langle$OUT,OUT,OUT,OUT$\rangle$};
            \draw (a1) -- (a2);
            \draw (a1) -- (a3);
            \draw (a2) -- (a4);
            \draw (a3) -- (a4);
        \end{tikzpicture}
    }
    \subfigure[]{
        \begin{tikzpicture}[scale=.7,every node/.style={scale=.6}]
            \node at (0,0) (a1) {$\langle$OUT,IN,IN,OUT$\rangle$};
            \node at (-1.5,-1) (a2) {$\langle$OUT,IN,IN,IN$\rangle$};
            \node at (1.5,-1) (a3) {$\langle$OUT,OUT,IN,OUT$\rangle$};
            \node at (0,-2) (a4) {$\langle$OUT,OUT,IN,IN$\rangle$};
            \draw (a1) -- (a2);
            \draw (a1) -- (a3);
            \draw (a2) -- (a4);
            \draw (a3) -- (a4);
        \end{tikzpicture}
    }
    \subfigure[]{
        \begin{tikzpicture}[scale=.7,every node/.style={scale=.6}]
            \node at (0,0) (a1) {$\langle$IN,OUT,OUT,IN$\rangle$};
            \node at (-1.5,-1) (a2) {$\langle$IN,IN,OUT,IN$\rangle$};
            \node at (1.5,-1) (a3) {$\langle$IN,OUT,OUT,OUT$\rangle$};
            \node at (0,-2) (a4) {$\langle$IN,IN,OUT,OUT$\rangle$};
            \draw (a1) -- (a2);
            \draw (a1) -- (a3);
            \draw (a2) -- (a4);
            \draw (a3) -- (a4);
        \end{tikzpicture}
    }
    \caption{The lattices induced by \Cref{algorithm:ds-ll} on the graph $G_4$ described in \Cref{example:dominating-set-definition}. Transitive edges are not shown for brevity.}
    \label{figure:full-lattices-from-ds-example}
\end{figure}

\section{Fully Lattice-Linear Algorithm for Graph Colouring (\gc)}\label{section:gc-ll}

In this section, we describe a fully lattice-linear algorithm for \gc. We first define the \gc problem, and then we describe an algorithm for \gc.

\begin{definition}\textbf{Graph colouring}. 
    In the \gc problem, the input is an arbitrary graph $G$ with some initial colouring assignment $\forall i\in V(G): i[colour]\in \mathbb{N}$. The task is to (re)assign the colour values of the nodes such that any adjacent nodes should not have a conflict (i.e., should not have the same colour), and there should not be a node whose colour can be reduced without conflict.
\end{definition}

We describe the algorithm as \Cref{algorithm:gc-ll}. Any node $i$, which has a conflicting colour with any of its neighbours, or if its colour value is reducible, is an \textit{unsatisfied} node. A node having a conflicting or reducible colour changes its colour to the lowest non-conflicting value iff it is \imped, i.e., any node $j$ in $Adj_i$ with ID more than $i$ is not unsatisfied. In the case that $i$ is \imped, if $i$ has a conflict with any of its neighbours, then we call it \textit{conflict-\imped}, if, otherwise, its colour is reducible, then we call it \textit{reducible-\imped}.

\begin{algorithm}\label{algorithm:gc-ll}Algorithm for \gc.
    \begin{center}
        \begin{tabular}{|l|}
            \hline 
            $\textsc{Conflicted-GC}(i)\equiv \exists j\in Adj_i:j[colour]=i[colour]$.\\
            $\textsc{Reducible-GC}(i)\equiv \exists c\in\mathbb{N},c<i[colour]:$\\ \quad\quad $(\forall j\in Adj_i:c\neq j[colour])$.\\
            \textsc{Unsatisfied-GC}$(i)$ $\equiv \textsc{Conflicted-GC}(i)\lor$\\ \quad\quad $ \textsc{Reducible-GC}(i)$.\\
            
            $\textsc{\Imped-GC}(i)\equiv\textsc{Unsatisfied-GC}(i)\land$
            ~\\
            \quad\quad $(\forall j\in Adj_i:\lnot\textsc{Unsatisfied-GC}(j)\lor$ $ i[id]>j[id])$.
            ~\\~\\
            Rules for node $i$.\\
            $\textsc{\Imped-GC}(i) \longrightarrow i[colour]=\min\{c\in \mathbb{N}:$\\ \quad\quad $ \forall j\in Adj_i, c\neq j[colour]\}$.\\
            \hline 
        \end{tabular}
    \end{center}
\end{algorithm}


    

\begin{lemma}\label{lemma:gc-decrease}
     Under \Cref{algorithm:gc-ll}, the colour value may increase or decrease at its first move, after which, its colour value monotonously decreases.
\end{lemma}

\begin{proof}

    
    When some node $i$ is deemed \imped for the first time, it may be conflicted or reducible. In either case, it obtains a colour value that is not conflicting with the colour value of its neighbours. The updated colour of $i$ will be a value from 1 to $|Adj_i|+1$. At this time, no neighbour of $i$ can change its colour. 
        
    Now, we show that $i$ will not become conflicted again, after becoming \imped for the first time. Under \Cref{algorithm:gc-ll}, any node $j$ in $Adj_i$ will not change its colour until it obtains the updated colour value of $i$. (If $j$ reads old information about $i[color]$ then it will continue to wait for $i$ to execute as required by the guard \textsc{\Imped-GC}.) If in case some node $j$ in $Adj_i$ becomes impedensable, then it must obtain a colour value that is not equal to the copy of $i[colour]$ that it reads/stores. Thus $i$ does not become conflicted by the action of $j$.

    Thus, after $i$ becomes \imped for the first time, it only reduces its colour in every subsequent move.
\end{proof}

To demonstrate that \Cref{algorithm:gc-ll} is lattice-linear, we define the state value and rank, auxiliary variables associated with nodes and global states, as follows.
    $$
    \begin{array}{l}
        \textsc{State-Value-GC}(i,s)=\\
        \begin{cases}
            deg(i)+2 & \text{if $\textsc{Conflicted-GC}(i)$ in state $s$} \\
            i[colour] & \text{otherwise}
        \end{cases}
    \end{array}
    $$
$$\textsc{Rank-GC}(s)=\sum\limits_{i\in V(G)}\textsc{State-Value-GC}(i,s).$$

\begin{theorem}\label{theorem:gc-ll}
    \Cref{algorithm:gc-ll} is a silent self-stabilizing and lattice-linear algorithm executed by $n$ nodes running asynchronously.
\end{theorem}

\begin{proof}
    From the proof of \Cref{lemma:gc-decrease}, we have that for any node $i$, $\textsc{State-Value-GC}(i)$ decreases when $i$ is \imped and never increases. This is because $i$ can increase its colour only once, after which it obtains a colour that is not in conflict with any of its neighbours, so any move that $i$ makes after that will reduce its colour. Therefore, $\textsc{Rank-GC}$ monotonously decreases until no node is \imped. This shows that \Cref{algorithm:gc-ll} is self-stabilizing. 
    
    In any suboptimal global state, at least one node is \imped, e.g., the highest ID node that is unsatisfied. Thus, a suboptimal global state will transition to a global state with a lesser rank. Since there are only a bounded number of colour values from 1 to $deg(i)+1$, a node can become reducible only a bounded number of times. If no node is conflicted and no node is reducible, and no node is \imped, no node makes a move. This shows that \Cref{algorithm:gc-ll} is silent.
    
    \Cref{algorithm:gc-ll} exhibits properties similar to \Cref{algorithm:ds-ll} which are elaborated in the proof for its lattice-linearity in \Cref{theorem:ds-ll}. Thus, \Cref{algorithm:gc-ll} is also lattice-linear.
\end{proof}

\section{Limitations of using Simple Actions and Tiebreakers for Developing FLLA}
\label{section:vc-no-ll}

We studied that lattice-linear algorithms for \mds and \gc can be designed by simply using tie-breakers.
Hence, a natural question arises if a lattice-linear algorithm can be designed for other graph theoretic problems by using some tie-breaker. The answer is no. Specifically, we cannot extend this design to develop algorithms for all graph theoretic problems -- we study minimal vertex cover (\mvc) and maximal independent set (\mis) problems in this context.

We first show (\Cref{subsection:vc-no-tie-breaker}) the issues involved in an algorithm that simply uses a tie-breaker to decide which node enters or leaves the vertex cover. Specifically, we show that this design results in cyclic behaviour.
Such behaviour is observed when we use \textit{simple actions}, where a node only changes the state of itself when it evaluates that its guards are true, with arbitrary-distance tie-breaker.zzzzzzzz
Similar results can be derived for the \mis problem.
Subsequently, in \Cref{subsection:mvc-ll} (respectively, in \Cref{subsection:mis-ll}), we show that a lattice-linear algorithm can be developed for \mvc (respectively, \mis) with \textit{complex actions}, where a node is allowed to make changes to the variables of other nodes. Then, in \Cref{subsection:complex-actions-vc-is}, we elaborate on the properties of algorithms, that we present, for \mvc and \mis.

\subsection{Issues in Using Only a Tie-Breaker in Algorithm for Minimal Vertex Cover (\mvc)}\label{subsection:vc-no-tie-breaker}


\begin{definition}\label{definition:mvc}\textbf{Minimal Vertex Cover.}
    In the MVC problem, the input is an arbitrary graph $G$, and the task is to compute a minimal set $\mathcal{V}$ such that for any edge $\{i,j\}\in E(G)$, $(i\in \mathcal{V})\lor (j\in \mathcal{V})$. If a node $i$ is in $\mathcal{V}$, then $i[st]=IN$, otherwise $i[st]=OUT$.
\end{definition}

We could use the macros $addable$ (some edge of a subject node is not covered) and $removable$ (removing the node preserves the vertex cover) to design an algorithm for \mvc. 
However, this design results in a cyclic behaviour, with respect to the local state transition of a node, even with a tie-breaker with all other nodes in the graph.
To illustrate this, consider the execution of an algorithm with such a tie-breaker on a line graph of 4 nodes (ID'd 1-4, sequentially) where all nodes are initialized to $OUT$.
Here, node 4 can change its state to $IN$. Other nodes cannot change their state because there is a node with a higher ID that can enter the vertex cover. After node 4 enters the vertex cover, node 3 enters the vertex cover, as edge $\{2, 3\}$ is not covered. However, this requires node 4 to leave the vertex cover to keep it minimal. 

Observe, above, that node $4$ was initialized such that $4[st]=OUT$, then it changed to $4[st]=IN$ and subsequently changed again to $4[st]=OUT$. 
Thus, we see a cyclic behaviour which is not desired in a lattice-linear algorithm. 
%
%
%
%
This analysis also shows that the use of simple actions results in the system exhibiting a cyclic behaviour. However, we have that complex actions can be utilized to move around this issue, which we study in the following.

\subsection{Fully Lattice-Linear Algorithm for Minimal Vertex Cover (\mvc)}\label{subsection:mvc-ll}

In \Cref{subsection:vc-no-tie-breaker}, we discussed the issues that arise in using (1) only a tie-breaker, and (2) simple actions. Based on these limitations, in this section, we describe a lattice-linear algorithm that utilizes complex actions to solve the \mvc problem. 

We use the following macros. 
A node $i$ is \textit{removable} iff $i$ is in the vertex cover, and all the neighbours of $i$ are also in the vertex cover.
$i$ is \textit{addable} iff $i$ is out of the vertex cover and there is some edge $\{i,j\}$ incident on $i$ such that $j$ is not in the vertex cover.
$i$ is \textit{unsatisfied} iff $i$ is removable or $i$ is addable.
$i$ is \textit{impedendable} iff $i$ is unsatisfied and there is no node $j$ in distance-3 of $i$, with ID greater than $i$, such that $j$ is unsatisfied.
\begin{center}
    \begin{tabular}{|l|}
        \hline 
        \textsc{Removable-VC}$(i)$ $\equiv i[st]=IN\land$\\ \quad\quad\quad\quad$(\forall j\in Adj_i: j[st]=IN)$.\\
        \textsc{Addable-VC}$(i)$ $\equiv i[st]=OUT\land$\\ \quad\quad\quad\quad$(\exists j\in Adj_i:j[st]=OUT)$.\\
        \textsc{Unsatisfied-VC}$(i)$ $\equiv \textsc{Removable-VC}(i)\lor$\\ \quad\quad\quad\quad$\textsc{Addable-VC}(i)$.\\
        $\textsc{\Imped-II-VC}(i)\equiv\textsc{Unsatisfied-VC}(i)\land$\\ \quad\quad\quad\quad $(\forall j\in Adj^3_i:\lnot\textsc{Unsatisfied-VC}(j)\lor$\\ \quad\quad\quad\quad$i[id]>j[id])$.\\
        \hline 
    \end{tabular}
\end{center}

The algorithm is defined as follows. An \textit{addable-\imped} node $i$ turns itself in and forces all its removable neighbours out. 
(after accounting for the fact that $i$ has already turned in). 
In this version of the algorithm, this complex action is assumed to be atomic. 
A \textit{removable-\imped} node $i$ will move out of the vertex cover. 

\begin{algorithm}\label{algorithm:vc-ll}Rules for node $i$.
\end{algorithm}
\centerline{
$\begin{array}{|l|}
    \hline 
    \textsc{\Imped-II-VC}(i) \longrightarrow\\
    \begin{cases}
        i[st]=IN.~\forall j\in Adj_i:\\ \quad\quad j[st]=OUT,\\ \quad\quad\text{if $\textsc{Removable}(j)$}. & \text{if $\textsc{Addable-VC}(i)$}.\\
        i[st] = OUT. & \text{otherwise}
    \end{cases}~\\
    \hline 
\end{array}$
}


\begin{lemma}\label{lemma:vc-no-step-back}
An addable (respectively, removable) node that enters (respectively, leaves) the vertex cover does not become removable (respectively addable) in any future time.

\end{lemma}

\begin{proof}
    Let $s$ be the state at time $t$ while \Cref{algorithm:vc-ll} is executing. We have from \Cref{algorithm:vc-ll} that if a node $i$ is addable-\imped or removable-\imped, then no other node in $Adj^3_i$ changes its state.

    If $i$ is removable-\imped at $t$, then all the nodes in $Adj_i$ are in the vertex cover. 
    Any node $j\in Adj_i$ cannot move out of the vertex cover until $i$ moves out. Hence, $i$ does not become addable after being removed at time $t$.

    If $i$ is addable-\imped at $t$, then some node in $Adj_i$ is out of the vertex cover. 
    After $i$ moves in and forces its removable neighbours out, $i$ is no longer addable and all nodes in $Adj_i$ are no longer removable.
    Now, $i$ can become removable only if some neighbour of $i$ enters the vertex cover. 
    Let $j$ be a neighbour of $i$. We consider two cases (1) $j[st]=OUT$ (2) $j[st]=IN$ and $j$ was forced out by $i$ (3) $j[st]=IN$ and $j$ was not forced out by $i$. 
    
    In the first case, as long as $j$ never enters the vertex cover, $i$ cannot become removable, thereby ensuring that edge $\{i, j\}$ remains covered. In the event $j$ enters the vertex cover, $i$ may be forced out. However, (as a result) $i$ does not become removable-\imped. 
    
    In the second case, $j$ is forced out of the vertex cover. This implies that all neighbours of $j$ are already in the vertex cover. Hence, $j$ does not become addable again. (This argument is the same as the above argument that showed that if a node is removed from the vertex cover, it does not become addable.) Since $j$ is never added back to the vertex cover, $i$ cannot become removable because of the action of $j$.
    
    In the third case, even if $j$ moves out afterwards, it cannot make $i$ removable, 
    so after time $t$, nodes in $Adj_i$ do not move in; as a result, $i$ does not have to move out after moving in. In addition, the nodes that $i$ turned out of the vertex cover do not have to move in after $t$, because the neighbours of those nodes are not changing their states simultaneously when their states are being changed.
    

    
    From the above cases, we have that $i$ does not change its state to $i[st]$ after changing its state from $i[st]$ to $i[st']$.
    throughout the execution of \Cref{algorithm:vc-ll}.
\end{proof}

To demonstrate that \Cref{algorithm:vc-ll} is lattice-linear, we define state value and rank, auxiliary variables associated with nodes and global states, as follows:
$$
    \begin{array}{l}
        \textsc{State-Value-VC}(i,s)=\\
        \begin{cases}
            1 & \text{if $\textsc{Unsatisfied-VC}(i)$ in state $s$} \\
            0 & \text{otherwise}
        \end{cases}
    \end{array}
$$
$$
    \textsc{Rank-VC}(s)=\sum\limits_{i\in V(G)}\textsc{State-Value-VC}(i,s).
$$



\begin{theorem}\label{theorem:vc-ll}
    \Cref{algorithm:vc-ll} is a silent self-stabilizing and lattice-linear algorithm executed by $n$ nodes running asynchronously.
\end{theorem}

\begin{proof}
    We have from the proof of \Cref{lemma:vc-no-step-back} that if $G$ is in state 
    $s$ and $\textsc{Rank-VC}(s)$ is non-zero, then at least one node will be \imped, e.g., an unsatisfied node with the highest ID.
    For any node $i$, we have that $\textsc{State-Value-VC}(i)$ decreases whenever $i$ is \imped and never increases. In addition, by the action of $i$, the state value of any other node in $G$ does not increase. (Note that adding node $i$ to the vertex cover may have caused a neighbour $j$ of $i$ to be removable. If this happens, $j$ is removed from the vertex cover in the same action atomically. Hence, $j$ does not become unsatisfied.)
    In effect, $\textsc{Rank-VC}$ monotonously decreases throughout the execution of the algorithm until it becomes zero. This shows that \Cref{algorithm:vc-ll} is self-stabilizing. Once \textsc{Rank-VC} is zero, no node is \imped, so no node makes a move. This shows that \Cref{algorithm:vc-ll} is silent.
    
    \Cref{algorithm:vc-ll} exhibits properties similar to \Cref{algorithm:ds-ll} (as well as \Cref{algorithm:gc-ll}) which are elaborated in the proof for its lattice-linearity in \Cref{theorem:ds-ll}. From there, we obtain that \Cref{algorithm:vc-ll} also is lattice-linear.
\end{proof}

\begin{example}\label{example:vc-semilattice}
    Let $G^p_4$ be a graph of four vertices forming a path $\langle v_1,v_2,v_3,v_4\rangle$ such that $v_i[id]>v_j[id]$ iff $i>j$. In \Cref{figure:vc-semilattice-example}, we show all possible state transitions that $G^p_4$ can go through under \Cref{algorithm:vc-ll}. The global states in the figure are of the form $\langle v_1[st],v_2[st],v_3[st],v_4[st]\rangle$.
    \qed 
\end{example}

\begin{figure}[ht]
    \centering
    
    \subfigure[]{
        \begin{tikzpicture}[scale=.7,every node/.style={scale=.7}]
            \node at (0,3) (a6) {$\langle OUT,IN,IN,OUT\rangle$};
            \node at (0,2) (a2) {$\langle OUT,OUT,IN,OUT\rangle$};
            \node at (0,1) (a1) {$\langle OUT,OUT,OUT,$\underline{$IN$}$\rangle$};
            \node at (0,0) (a0) {$\langle OUT,OUT,OUT,OUT\rangle$};
            
            \node at (4,1) (a3) {$\langle OUT,OUT,IN,IN\rangle$};
            
            \node at (4,2) (a7) {$\langle OUT,IN,IN,IN\rangle$};

            \draw (a0) -- (a1) -- (a2) -- (a6);
            \draw (a3) -- (a2);
            \draw (a7) -- (a6);
        \end{tikzpicture}
    }
    \subfigure[]{
        \begin{tikzpicture}[scale=.7,every node/.style={scale=.7}]
            \node at (0,5) (a5) {$\langle OUT,IN,OUT,IN\rangle$};
            \node at (0,4) (a4) {$\langle OUT,IN,OUT,OUT\rangle$};
            
            \node at (4,4) (a13) {$\langle IN,IN,OUT,IN\rangle$};
            \node at (4,3) (a12) {$\langle IN,IN,OUT,OUT\rangle$};

            \draw (a4) -- (a5);
            \draw (a12) -- (a13) -- (a5);
        \end{tikzpicture}
    }
    \subfigure[]{
        \begin{tikzpicture}[scale=.7,every node/.style={scale=.7}]
            \node at (4,10) (a10) {$\langle IN,OUT,IN,OUT\rangle$};
            \node at (0,9) (a9) {$\langle IN,OUT,OUT,$\underline{$IN$}$\rangle$};
            \node at (0,8) (a8) {$\langle IN,OUT,OUT,OUT\rangle$};
            
            \node at (4,9) (a11) {$\langle IN,OUT,IN,IN\rangle$};
            
            \node at (8,9) (a14) {$\langle IN,IN,IN,OUT\rangle$};
            \node at (8,8) (a15) {$\langle IN,IN,IN,IN\rangle$};
            
            \draw (a8) -- (a9) -- (a10);
            \draw (a11) -- (a10);
            \draw (a15) -- (a14) -- (a10);
        \end{tikzpicture}
    }
    \caption{$\prec$-lattices formed by the global states of $G^p_4$ of 4 nodes forming a straight path $\langle v_1,v_2,v_3,v_4\rangle$ under \Cref{algorithm:vc-ll}. The nodes that are kicked out, by a node that decides to move into the vertex cover, are underlined.}
    \label{figure:vc-semilattice-example}
\end{figure}

\subsection{Fully Lattice-Linear Algorithm for Maximal Independent Set (\mis)}\label{subsection:mis-ll}

In this subsection, we describe an algorithm for the \mis. The issues, similar to the issues that we discussed in \Cref{subsection:vc-no-tie-breaker}, can be observed for \mis problem as well. For example, a similar behaviour can be observed on a path of 4 nodes, all initialized to $IN$. However, we can follow the general design of \Cref{algorithm:vc-ll}, that we described in \Cref{subsection:mvc-ll}, to develop an algorithm for \mis. First, we define the \mis problem as follows.

\begin{definition}\label{definition:mis}\textbf{Maximal Independent Set.}
    In the \textit{maximal independent set} (\mis) problem, the input is an arbitrary graph $G$, and the task is to compute a maximal set $\mathcal{I}$ such that for any two nodes $i\in\mathcal{I}$ and $j\in\mathcal{I}$, if $i\neq j$, then $\{i,j\}\neq E(G)$.
\end{definition}

The macros that we use here are similar to the macros we used for \mvc, but with opposite polarity. We use the following macros. 
A node $i$ is \textit{addable} iff $i$ is out of the independent set (\is), and all the neighbours of $i$ are also out of the \is.
$i$ is \textit{removable} iff $i$ is in the \is and there is some edge $\{i,j\}$ incident on $i$ such that $j$ is also in the \is.
$i$ is \textit{unsatisfied} iff $i$ is removable or $i$ is addable.
$i$ is \textit{\imped} iff $i$ is unsatisfied and there is no node $j$ in distance-3 of $i$, with ID greater than $i$, such that $j$ is unsatisfied.

\begin{center}
    \begin{tabular}{|l|}
        \hline 
        \textsc{Addable-IS}$(i)$ $\equiv i[st]=OUT\land$\\ \quad\quad\quad\quad$(\forall j\in Adj_i j[st]=OUT)$.\\
        \textsc{Removable-IS}$(i)$ $\equiv i[st]=IN\land$\\ \quad\quad\quad\quad$(\exists j\in Adj_i:j[st]=IN)$.\\
        \textsc{Unsatisfied-IS}$(i)$ $\equiv \textsc{Removable-IS}(i)\lor$\\ \quad\quad\quad\quad$\textsc{Addable-IS}(i)$.\\
        $\textsc{\Imped-II-IS}(i)\equiv\textsc{Unsatisfied-IS}(i)\land$\\ \quad\quad\quad\quad $(\forall j\in Adj^2_i:\lnot\textsc{Unsatisfied-IS}(j)\lor$\\ \quad\quad\quad\quad$i[id]>j[id])$.\\
        \hline 
    \end{tabular}
\end{center}

Given the similarities in the \mvc and \mis problems, the algorithm we develop here is similar to the algorithm we developed for \mvc. The algorithm is defined as follows. A \textit{removable-\imped} node $i$ turns itself out and moves all its addable neighbours into the independent set (after accounting for the fact that $i$ has already turned out). An \textit{addable-\imped} node $i$ will turn itself into the independent set. In the present design of the algorithm, the execution while $i$ is \imped is assumed to be executed atomically.

\begin{algorithm}\label{algorithm:is-ll}Rules for node $i$.
\end{algorithm}
\begin{center}
    $\begin{array}{|l|}
        \hline 
        \textsc{\Imped-II-IS}(i) \longrightarrow\\
        \begin{cases}
            i[st]=OUT.~\forall j\in Adj_i:\\ \quad\quad j[st]=IN,\\ \quad\quad\text{if $\textsc{Addable}(j)$}. & \text{if $\textsc{Removable-IS}(i)$}.\\
            i[st] = IN. & \text{otherwise}
        \end{cases}~\\
        \hline 
    \end{array}$
\end{center}

Since the behaviour of \Cref{algorithm:is-ll} is similar to the behaviour of \Cref{algorithm:vc-ll}, we briefly cover the description of the behaviour of \Cref{algorithm:is-ll} in the following: \Cref{lemma:is-no-step-back} and \Cref{theorem:is-ll}.

\begin{lemma}\label{lemma:is-no-step-back}
    Any node in an input graph does not revisit its older state while executing under \Cref{algorithm:is-ll}.
\end{lemma}

\begin{proof}
    Let $s$ be the state at time $t$ while \Cref{algorithm:is-ll} is executing. We have from \Cref{algorithm:is-ll} that if a node $i$ is addable-\imped or removable-\imped, then no other node in $Adj^3_i$ changes its state.
    
    If $i$ is removable-\imped at $t$, then some node in $Adj_i$ is in the independent set. After when $i$ moves out, and turns its addable neighbours in, then we have that $i$ is no longer removable and all nodes in $Adj_i$ are no longer addable, so after time $t$, nodes in $Adj_i$ do not move out; as a result, $i$ does not have to move in after moving out. In addition, the nodes that $i$ turned into the independent set do not have to move out after $t$, because the neighbours of those nodes are not changing their states simultaneously when their states are being changed.
    
    Let that at some instance of time, two nodes $i$ and $j$ simultaneously evaluate that they are removable-\imped. Since no nodes in $Adj_i^3$ (respectively, $Adj_{j}^3$) change their state under \Cref{algorithm:is-ll} until $i$ (respectively, $j$) changes its state, we have that $i$ and the nodes that $i$ would move into the independent set are not adjacent to $j$ or the nodes that $j$ would move into the independent set. Thus, no node by the action of $i$ becomes removable.
    
    If otherwise $i$ is addable-\imped at $t$, then all the nodes in $Adj_i$ are out of \is. So after when $i$ moves in, none of the non-unsatisfied nodes in $Adj_i$ (non-addable nodes in $Adj_i$), including $i$, becomes unsatisfied (addable). As a result, $i$ does not have to move out after moving in.
    
    From the above cases, we have that $i$ does not change its state to $i[st]$ after changing its state from $i[st]$ to $i[st']$.
    throughout the execution of \Cref{algorithm:is-ll}.
\end{proof}

To demonstrate that \Cref{algorithm:is-ll} is lattice-linear, we define state value and rank, auxiliary variables associated with nodes and global states, as follows:
$$
    \begin{array}{l}
        \textsc{State-Value-IS}(i,s)=\\
        \begin{cases}
            1 & \text{if $\textsc{Unsatisfied-IS}(i)$ in state $s$} \\
            0 & \text{otherwise}
        \end{cases}
    \end{array}
$$
$$
    \textsc{Rank-IS}(s)=\sum\limits_{i\in V(G)}\textsc{State-Value-IS}(i,s).
$$



\begin{theorem}\label{theorem:is-ll}
    \Cref{algorithm:is-ll} is a silent self-stabilizing and lattice-linear algorithm executed by $n$ nodes running asynchronously.
\end{theorem}

\begin{proof}
    We have from the proof of \Cref{lemma:is-no-step-back} that if $G$ is in state 
    $s$ and $\textsc{Rank-IS}(s)$ is non-zero, then at least one node will be \imped, e.g., an unsatisfied node with the highest ID.
    For any node $i$, we have that $\textsc{State-Value-IS}(i)$ decreases whenever $i$ is \imped and never increases. In addition, by the action of $i$, the state value of any other node in $G$ does not increase.
    In effect, $\textsc{Rank-IS}$ monotonously decreases throughout the execution of the algorithm until it becomes zero. This shows that \Cref{algorithm:is-ll} is self-stabilizing.
    Once \textsc{Rank-DS} is zero, no node is \imped, so no node makes a move. This shows that \Cref{algorithm:is-ll} is silent.
    
    \Cref{algorithm:is-ll} exhibits properties similar to \Cref{algorithm:ds-ll} (as well as \Cref{algorithm:gc-ll} and \Cref{algorithm:vc-ll}) which are elaborated in the proof for its lattice-linearity in \Cref{theorem:ds-ll}. From there, we obtain that \Cref{algorithm:is-ll} also is lattice-linear.
\end{proof}

\subsection{Complex Actions: Properties Shared by Algorithms for \mvc and \mis}\label{subsection:complex-actions-vc-is}

In this subsection, we study some behavioural aspects   of the algorithm for \mvc present in \Cref{subsection:mvc-ll}. Consequently, similar arguments for the algorithm for \mis will follow.


\Cref{algorithm:vc-ll} can be transformed into a simple action algorithm as follows. To accommodate that, we use the variable $i[addable]$ to set to be $true$ so that the surrounding nodes can then evaluate if they are removable. We use the following additional guards.
For a node $i$, \textit{else-pointed} is true iff a node $j$ in $Adj^4_i$ moved into the \vc (and set $j[addable]$ to $true$), and there is a node $k$ in $Adj_j$ that is removable.

\begin{center}
        $\textsc{Else-Pointed}(i)\equiv \exists j\in Adj^4_i:(j[addable=true\land \exists k\in Adj_j: \textsc{Removable}(k)])$
\end{center}

Consequently, the algorithm can be modified as follows. A node $i$ will not be enabled (\imped) if else-pointed is true for $i$. The modified algorithm that allows simple actions, thus, is defined as follows.

\begin{algorithm}\label{algorithm:vc-ll-simple-actions}
    Transformed \Cref{algorithm:vc-ll}, where nodes only execute simple actions.
\end{algorithm}
\begin{center}
    $\begin{array}{|l|}
        \hline 
        \textsc{Unsatisfied-II-VC}(i)\equiv \textsc{Removable-VC}(i)\lor\\ \quad\quad\quad\quad\textsc{Addable-VC}(i).\\
        \textsc{\Imped-II-VC}(i)\equiv\\
        (\exists j\in Adj_i:j[addable]=true \land \textsc{Removable}(i))\lor\\
        \quad\quad\quad\quad (\lnot\textsc{Else-Pointed}(i)\land (\textsc{Unsatisfied-VC}(i)\land\\ 
        \quad\quad\quad\quad (\forall j\in Adj^3_i:\lnot\textsc{Unsatisfied-VC}(j)\lor\\ \quad\quad\quad\quad i[id]>j[id]))).\\~\\
        \text{Rules for node $i$}:\\
        \textsc{\Imped-VC-II}(i)\longrightarrow\\
        \begin{cases}
            i[addable]=true & \text{if $i[st]=OUT$.}\\
            i[addable]=false & \text{if $i[st]=IN$.}\\
            i[st]=\lnot i[st] & \text{unconditionally.}\\
        \end{cases}
        ~\\
        \hline 
    \end{array}$
\end{center}

Next, we identify why \Cref{algorithm:vc-ll-simple-actions} can be reconciled with the inability to design an algorithm with simple actions from \Cref{subsection:vc-no-tie-breaker}. This analysis also helps us to obtain the correctness proof of \Cref{algorithm:vc-ll-simple-actions}.

\paragraph{Reconciling \Cref{subsection:vc-no-tie-breaker} and \Cref{algorithm:vc-ll-simple-actions}}

As discussed in \Cref{subsection:vc-no-tie-breaker}, if a tie-breaker in conjunction with simple actions is deployed, then the system would exhibit cyclic behaviour. On the other hand, Algorithm \ref{algorithm:vc-ll-simple-actions} uses only simple actions. We identify the subtlety, involved in both these results, that make this possible, in the following.


To explain how these results can coexist together, we describe the behaviour of \Cref{algorithm:vc-ll-simple-actions} with an example. Assume that this algorithm is deployed on a graph of 4 nodes forming a path, with node IDs being in the sequence $\langle 1,4,3,2\rangle$. In a simple algorithm that uses a tiebreaker with all nodes (which is considered in \Cref{subsection:vc-no-tie-breaker}), node 4 would execute first then node 3, then node 2 and finally node 1. This execution order is not preserved in \Cref{algorithm:vc-ll-simple-actions}, which we discuss as follows.

Specifically, let the initial global state in this graph be $\langle IN,OUT,OUT,OUT\rangle$. First, node 4 will move into the \vc. Now, the node that is unsatisfied and has the highest ID is node 3. However, \textsc{Else-Pointed}(3) is true, because $4[addable]$ is set to $true$ and node $1$
needs to move out of the vertex cover as part of the action that allowed node $4$ to enter the vertex cover.
In other words, node $3$ can execute only after node $1$ leaves the vertex cover. This is not permitted in a algorithm that uses simple actions with tie-breaker on node IDs. This effect is similar to that of priority inheritance, where node $1$ inherits the priority of node $4$ because it has to be forced out of the vertex cover by node $4$. Hence, in this specific case, 
node 1 has a higher priority of movement as compared to node 3, despite the fact that node 3 is unsatisfied and is of a higher ID, because of a recent action committed by node 4.

After node 1 moves out, node 3 moves in, and thence, the system reaches an optimal state. 

\paragraph{Correctness of \Cref{algorithm:vc-ll-simple-actions}}

\Cref{algorithm:is-ll} is lattice-linear with respect to state value and rank, defined as follows.
$$
    \begin{array}{l}
        \textsc{State-Value-II-VC}(i,s)=\\
        \begin{cases}
            |Adj_i+1|\\ \quad\quad \text{if $\textsc{Unsatisfied-II-VC}(i)$ in state $s$} \\
            |\{j\in Adj_i: \textsc{Unsatisfied-II-VC}(j)\}|\\ \quad\quad \text{if in state $s$, $\lnot\textsc{Unsatisfied-II-VC}(i)\land$}\\
            \quad\quad\quad\quad \text{$(\exists j\in Adj_i: \textsc{Unsatisfied-II-VC}(j))$} \\
            0\\ \quad\quad \text{otherwise}
        \end{cases}
    \end{array}
$$
$$
    \textsc{Rank-II-VC}(s)=\sum\limits_{i\in V(G)}\textsc{State-Value-II-VC}(i,s).
$$

Since the working of this algorithm straightforwardly follows from the working of \Cref{algorithm:vc-ll}, we omit the proof of correctness of this algorithm. A similar algorithm can be developed for \mis, that deploys only simple actions.




\section{Time Complexity: Lattice-Linear Algorithms}\label{section:convergence-time}


\begin{theorem}\label{theorem:general-convergence-time}
    Given a system of $n$ processes, with the domain of size not more than $m$ for each process, the acting algorithm will converge in $n\times (m-1)$ moves.
\end{theorem}

\begin{proof}
    Assume for contradiction that the underlying algorithm converges in $x\geq n\times (m-1)+1$ moves. This implies, by pigeonhole principle, that at least one of the nodes $i$ is revisiting its state $i[st]$ after changing to $i[st']$. If $i[st]$ to $i[st']$ is a step ahead transition for $i$, then $i[st']$ to $i[st]$ is a step back transition for $i$ and vice versa. For a system where the global states form a $\prec$-lattice, we obtain a contradiction since step back actions are absent in such systems.
\end{proof}

\begin{exampledscont}\label{example:number-of-moves}
    Consider phase 2 of \Cref{algorithm:ds-ellss}. 
    As discussed earlier, this phase is lattice-linear. The domain of each process $\{IN$, $OUT\}$ is of size 2. Hence, phase 2 of \Cref{algorithm:ds-ellss} requires at most $n\times (2-1)=n$ moves. (Phase 1 also requires atmost $n$ moves. But this fact is not relevant with respect to \Cref{theorem:general-convergence-time}.)
     
\end{exampledscont}

\begin{examplesmpcont}
    Observe from \Cref{figure:smplattice} that any system of 3 men and 3 women with arbitrary preference lists will converge in $3\times (3-1)=6$ moves. This comes from 3 men (resulting in 3 processes) and 3 women (domain size of each man (process) is 3).
\end{examplesmpcont}

\begin{corollary}\label{corollary:convergence-time-ll-multivariable}
    Let that the nodes are multivariable, and atmost $r$ variables $i[var_1],...,i[var_r]$ (with domain sizes $m_1',...m_r'$ respectively) contribute to the total order. 
    Such a system converges in $n\times \bigg(\Big(\prod\limits_{j=1}^r m_j'\Big)-1\bigg)$ moves.
     
\end{corollary}

\noindent\textbf{\textit{Remark}}: In \cite{Gupta2022}, this upper bound is $n\times \Big(\sum\limits_{j=1}^r (m_j'-1)\Big)$. This would be correct if the variables individually do not visit their value again.
However, the size of the local state can be as large as the product of the domain sizes of the individual variables. Then, if the local states form a total order, then the expression in the above corollary is applicable.

\begin{corollary}
    (From \Cref{theorem:ds-ll} and \Cref{theorem:general-convergence-time}) \Cref{algorithm:ds-ll} converges in $n$ moves.
\end{corollary}

\begin{corollary}
    (From \Cref{theorem:gc-ll} and \Cref{theorem:general-convergence-time}) \Cref{algorithm:gc-ll} converges in $\sum\limits_{i\in V(G)}deg(i)+1=n+2m$ moves.
\end{corollary}

\begin{proof}
    This can be reasoned as follows: first, a node may increase its colour, once, to resolve its colour conflict with a neighbouring node. Then it will decrease its colour, whenever it moves. Depending on the colour value of its neighbours and when they decide to move, a node $i$ can decrease its colour almost $deg(i)$ times.
\end{proof}

\begin{corollary}
    (From \Cref{theorem:vc-ll} and \Cref{theorem:general-convergence-time}) \Cref{algorithm:vc-ll} converges in $n$ moves.
\end{corollary}

\begin{corollary}
    (From \Cref{theorem:is-ll} and \Cref{theorem:general-convergence-time}) \Cref{algorithm:is-ll} converges in $n$ moves.
\end{corollary}

\section{Related Work}\label{section:literature}

\textbf{Lattice-linearity}: In \cite{Garg2020}, the authors have studied lattice-linear problems which possess a predicate under which the states naturally form a lattice among all states. Problems like the stable marriage problem, job scheduling
and others are studied in \cite{Garg2020}. In \cite{Garg2021} and \cite{Garg2022}, the authors have studied lattice-linearity in, respectively, housing market problem and several dynamic programming problems.

We studied lattice-linear problems that allow self-stabilization in \cite{Gupta2023a}. Specifically, we found that parallelized versions of multiplication and modulo are lattice-linear, and that they allow self-stabilization. In \cite{Gupta2021}, we introduced eventually lattice-linear algorithms. We developed eventually lattice-linear self-stabilizing algorithms for some non-lattice-linear problems, which impose a lattice in a subset of the state space.

In this paper, we introduce fully lattice-linear algorithms and present algorithms for some non-lattice-linear problems. These algorithms induce a lattice among all reachable states.

Garg, in \cite{Garg2020} studied problems in which a distributive lattice is formed among the global states, where a meet and join can be found for any given pair of states, and meet and join distribute over each-other. However, we, in this paper, find that to allow asynchrony, a more relaxed data structure can be allowed. Specifically, in a $\prec$-lattice, for a pair of global states, their join can be found, however, their meet may not be found. As an illustration, in the instance that we study in \Cref{figure:full-lattices-from-ds-example}, both meet and join can be found for a pair of global states in a $\prec$-lattice, however, in the instance that we study in \Cref{figure:vc-semilattice-example}, a join can be found for a pair of global states in a $\prec$-lattice but a meet is not always found.

\noindent\textbf{Dominating set}: Self-stabilizing algorithms for the minimal dominating set problem have been proposed in several works in the literature, for example, in \cite{Xu2003,GODDARD2008,Chiu2014}. Apart from these, the algorithm in \cite{Hedetniemi2003} converges in $O(n^2)$ moves, and the algorithm in \cite{Turau2007} converges in $9n$ moves under an unfair distributed scheduler. The best convergence time among these works is $4n$ moves. The eventually lattice-linear algorithm presented in \cite{Gupta2021} for a more generalized version, i.e., the service demand based \mds problem, takes $2n$ moves to converge.

In this paper, the fully lattice-linear algorithm that we present converges in $n$ moves and is fully tolerant to asynchrony. This is an improvement over the algorithms presented in the literature.

\noindent\textbf{Colouring}: Self-stabilizing algorithms for graph colouring have been presented in several works, including  \cite{Bhartia2016,Checco2017,Duffy2013,Duffy2008,Galan2017,Leith2006,Motskin2009,Chakrabarty2020,Gupta2021}. The best convergence time among these algorithms is $n\times \Delta$ moves, where $\Delta$ is the maximum degree of the input graph.

The fully lattice-linear algorithm for graph colouring that we present converges in $n+2m$ moves and is fully tolerant to asynchrony.

\noindent \textbf{Vertex Cover}: Self-stabilizing algorithms for the vertex cover problem
has been studied in Astrand and Suomela (2010) \cite{Astrand2010}, that converges in $O(\Delta)$ rounds, and Turau (2010) \cite{Turau2010}, that converges in $O(\min\{n, \Delta^2, \Delta \log_3 n\})$ rounds. The eventually lattice-linear algorithm that we presented in \cite{Gupta2021} converges in $2n$ moves in asynchrony.

The algorithm present in this paper is self-stabilizing, converges in $n$ moves, and is fully tolerant to asynchrony.

\noindent \textbf{Independent Set}: Self-stabilizing algorithm for maximal independent set has been presented in \cite{Turau2007}, that converges in $\max\{3n-5, 2n\}$ moves under an unfair distributed scheduler, \cite{GODDARD2008}, that converges in $n$ rounds under a distributed or synchronous scheduler, \cite{Hedetniemi2003}, that converges in $2n$ moves. The eventually lattice-linear algorithm that we presented in \cite{Gupta2021} converges in $2n$ moves in asynchrony.

The algorithm present in this paper is self-stabilizing, converges in $n$ moves, and is fully tolerant to asynchrony.

\noindent\textbf{Tie Breaker}: Most distributed algorithms need a tiebreaker to establish the dominance of actions of some nodes over others selectively. The algorithms for dominating set, colouring, vertex cover, independent set and others in the literature have used the node ID as the tiebreaker. In this paper, we, as well, have used node ID as a tiebreaker. In place of node ID, any other tie-breaker (e.g. degree of nodes) can be deployed.

While any algorithm in the problems discussed in this paper would require some tie-breaker, the algorithms in the literature require some synchronization to ensure convergence. On the other hand, the algorithms that we present in this paper do not require any synchronization primitives.

\noindent\textbf{Abstractions in Concurrent Computing}:
Since this paper focuses on asynchronous computations, we also study other abstractions in the context of concurrent systems: non-blocking (lock-free/wait-free), starvation-free and serializability. A comprehensive discussion of these models is present in the Appendix.

\section{Experiments}\label{section:experiments}

In this section, we present the experimental results of convergence times from implementations run on real-time shared memory model. We implement the algorithm for minimal dominating set (\Cref{algorithm:ds-ll}), and compare it to algorithms by Hedetniemi et al. (2003) \cite{Hedetniemi2003} and Turau (2007) \cite{Turau2007}. We also present the runtime of a distance-1 transformation of \Cref{algorithm:ds-ll}. First, we present the transformation of \Cref{algorithm:ds-ll} in the following subsection.

\subsection{Transforming \Cref{algorithm:ds-ll} to distance-1}

In \Cref{algorithm:ds-ll}, we observe that the guards of a node $i$ are distance-4. First, we transform this algorithm to a distance-1 algorithm. To accomplish this, the nodes maintain additional variables, that provide them information about other nodes, as required.
We use additional variables and guards to propagate this information.
Due to the constraint of reading only distance-1 neighbours, the nodes may end up reading old information about the other nodes. However, due to lattice-linearity, such executions stay to be correct.

A straightforward transformation would require each node $i$ to maintain copies of all the variables of its distance-4 neighbours. However, we use only 4 additional variables. Note that the requirement of these four variables is independent of the number of nodes in $Adj_i^4$. 

We use $i[ldom]$ and $i[hdom]$ to assist in propagating the information about the macro $\textsc{Removable-DS}(i)$.
$i[hdom]$ stores the highest ID dominator of $i$: it is a node in $N_i$ of highest ID such that its state is $IN$. $i[ldom]$ stores the lowest ID dominator of $i$: it is a node in $N_i$ of lowest ID such that its state is $IN$. 
(If such a node does not exist, these variables are set to $\top$ (\textit{null}).)

Once Removable-DS is transformed to a distance-1 macro, unsatisfied-ds will also be a distance-1 macro, as $\textsc{Addadble-DS}$  is already a distance-1 macro. 

We use $i[uflag]$ and $i[hud1]$ to assist in propagating the information about $\textsc{\Imped-II-DS}$.
Node $i$ sets $i[uflag]$ to $true$ to indicate that $i$ is unsatisfied.
$i[hud1]$ stores the highest ID node in distance-1, i.e., in $N_i$, of $i$ that is unsatisfied.


Now, we describe the actions of the transformed distance-1 algorithm. 
We use the following macros. $i$ is \textit{hdom-outdated} iff $i[hdom]$ is not equal to the highest ID dominator of $i$. $i$ is \textit{ldom-outdated} iff $i[ldom]$ is not equal to the lowest ID dominator of $i$. 
$i$ is \textit{removable} iff every node $j \in N_i$ will stay dominated even if $i$ moves out of the dominating set. This will happen if either $j[st]=IN$, or, either $j[hdom]$ or $j[ldom]$ differs from $i$. 
Node $i$ is \textit{addable} if all nodes in $Adj_i$, along with $i$, are out of the \ds. 
$i$ is \textit{unsatisfied} if $i$ is removable or addable. $i$ is \textit{unsatisfied-flag-outdated} iff $i[uflag]$ is not equivalent to $i$ being unsatisfied. $i$ is \textit{hud1-outdated} iff $i[hud1]$ is not qual to the highest ID node in $N_i$ that is unsatisfied. $i$ is \textit{unsatisfied-\imped} if $i$ is the highest ID node in the distance-2 neighbourhood that is unsatisfied. $i$ is \textit{\imped} iff $i$ is hdom-outdated, ldom-outdated, unsatisfied-flag-outdated, hud1-outdated or unsatisfied-\imped.

\begin{center}
    \begin{tabular}{|l|}
        \hline 
        variables of $i$: $st,ldom,hdom,uflag,hud1$\\
        $\textsc{HDom-Outdated}(i)\equiv i[hdom]\neq \text{arg} \max\{x[id]:$\\ \quad\quad$x\in N_i\land x[st]=IN\}$.\\
        $\textsc{LDom-Outdated}(i)\equiv i[ldom]\neq \text{arg} \min\{x[id]:$\\ \quad\quad$x\in N_i\land x[st]=IN\}$.\\
        $\textsc{Removable-\ds-D}\equiv i[st]=IN\land (\forall j\in N_i:$\\ \quad\quad$((j\neq i\land j[st]=IN)\lor ((j[ldom]\neq i\land$\\ \quad\quad$j[ldom]\neq \top)\lor(j[hdom]\neq i\land$\\ \quad\quad$j[hdom]\neq \top))))$.\\
        $\textsc{Addable-\ds-D}(i) \equiv i[st]=OUT\land$\\ \quad\quad$(\forall j\in Adj_i:j[st]=OUT)$.\\
        $\textsc{Unsatisfied-\ds-D}(i)\equiv \textsc{Removable-DS-D}(i)\lor$\\ \quad\quad$\textsc{Addable-DS-D}(i)$.\\
        $\textsc{Unsatisfied-Flag-Outdated}(i)\equiv$\\ \quad\quad$i[uflag]\not\equiv \textsc{Unsatisfied-\ds-D}(i)$.\\
        $\textsc{HUD1-Outdated}(i)\equiv i[hud1]\neq \text{arg} \max\{x[id]:$\\ \quad\quad$x\in N_i\land x[uflag]=true\}$.\\
        $\textsc{Unsatisfied-\Imped}(i)\equiv i[uflag]\land$\\ \quad\quad $(\forall j\in Adj_i:(j[uflag]\Rightarrow j[id]<i[id])\land$\\ \quad\quad$(j[hud1]\neq \top\Rightarrow j[hud1]<i[id]))$.\\
        $\textsc{\Imped-\ds-D}(i)\equiv \textsc{HDom-Outdated}(i)\lor$\\ \quad\quad$\textsc{LDom-Outdated}(i)\lor$\\ \quad\quad $\textsc{Unsatisfied-Flag-Outdated}(i)\lor$\\ \quad\quad$\textsc{HUD1-Outdated}(i)\lor$\\ \quad\quad $\textsc{Unsatisfied-\Imped}(i)$.\\
        \hline 
    \end{tabular}
\end{center}

We describe the algorithm as follows. If $i$ is hdom-outdated, then it updates $i[hdom]$ with ID of the highest ID node in $N_i$ that is in the dominating set. If $i$ is ldom-outdated, then it updated $i[ldom]$ with ID of the lowest ID node in $N_i$ that is in the dominating set. If $i$ is unsatisfied-flag-outdated, then it updates $i[uflag]$ to correctly denote whether $i$ is unsatisfied. If $i$ is hud1-outdated, then $i$ updates $i[hud1]$ with the ID of the highest ID node in $N_i$ that is unsatisfied. If $i$ is unsatisfied-impedensable, then $i$ toggles $i[st]$.

\begin{algorithm}\label{algorithm:ds-ll-d1}
    Rules for node $i$.
\end{algorithm}
\begin{center}
    $\begin{array}{|l|}
        \hline 
        \textsc{\Imped-\ds-M}(i)\longrightarrow\\
        \begin{cases}
            hdom=\text{arg} \max\{x[id]:x\in N_i\land x[st]=IN\}\\ \quad\quad\quad \text{if $\textsc{HDom-Outdated}(i)$}.\\
            ldom=\text{arg} \min\{x[id]:x\in N_i\land x[st]=IN\}\\ \quad\quad\quad \text{if $\textsc{LDom-Outdated}(i)$}.\\
            uflag\not\equiv \textsc{Unsatisfied-\ds-M}(i)\\ \quad\quad\quad \text{if $\textsc{Unsatisfied-Flag-\Imped}(i)$}.\\
            i[hud1]=\text{arg} \max\{x[id]:\\ \quad x\in N_i\land x[uflag]=true\}\\ \quad\quad\quad \text{if $\textsc{HUD1-Outdated}(i)$}.\\
            i[st]=\lnot i[st]\\ \quad\quad\quad \text{if $\textsc{Unsatisfied-\Imped}(i)$}.\\
            i[uflag]=false\\ \quad\quad\quad \text{unconditionally}.
        \end{cases}~\\
        \hline 
    \end{array}$
\end{center}

Deploying the above algorithm reduces the work complexity of evaluating the guards for a node to $O(\Delta)$, which was originally $O(\Delta^4)$ in \Cref{algorithm:ds-ll}. In principle, all lattice-linear distance-$x$ (where $x>1$) algorithms can be transformed into distance-1 algorithms by 
keeping a copy of all variables in $Adj^x$ \cite{Afek2002}. However, it will increase the space complexity of every node by $|Adj^x|$, without decreasing the time complexity of the evaluation of guards.
By contrast, the above algorithm increases the space complexity by only $O(1)$, while decreasing the time complexity from $O(\Delta^4)$ to $O(\Delta)$.

\subsection{Runtime Comparison}

While we see a significant reduction of the time complexity, of the evaluation of guards by a node, from $O(\Delta^4)$ in \Cref{algorithm:ds-ll} to $O(\Delta)$ in \Cref{algorithm:ds-ll-d1}, it is also worthwhile to compare the convergence time of these algorithms when they are implemented on real-time systems.
In this subsection, we compare the runtime of \Cref{algorithm:ds-ll-d1} with \Cref{algorithm:ds-ll} and other algorithms.

\begin{figure}[ht]
    \centering
    \subfigure[]{
        \includegraphics[width=.44\textwidth]{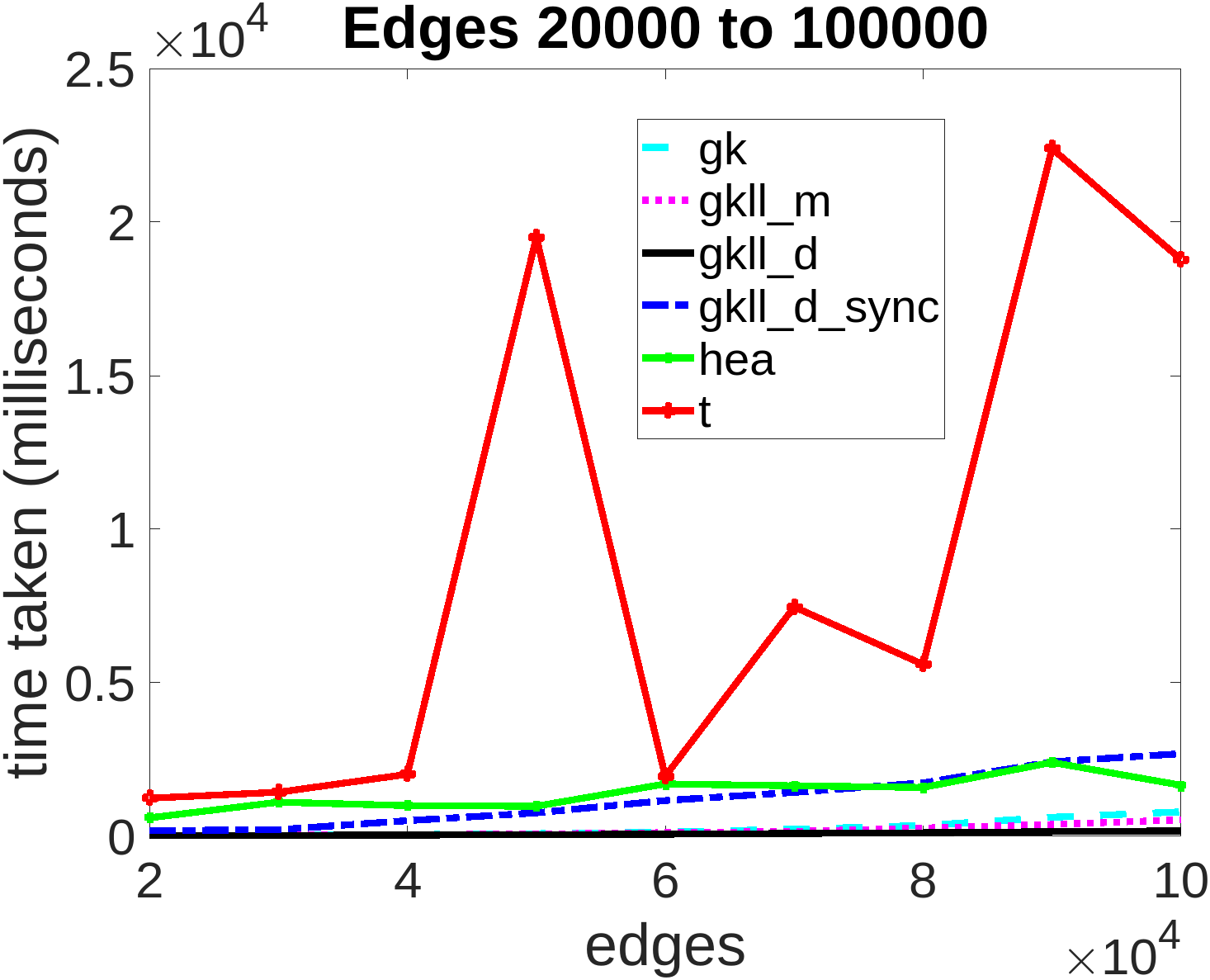}
    }
    \subfigure[]{
        \includegraphics[width=.44\textwidth]{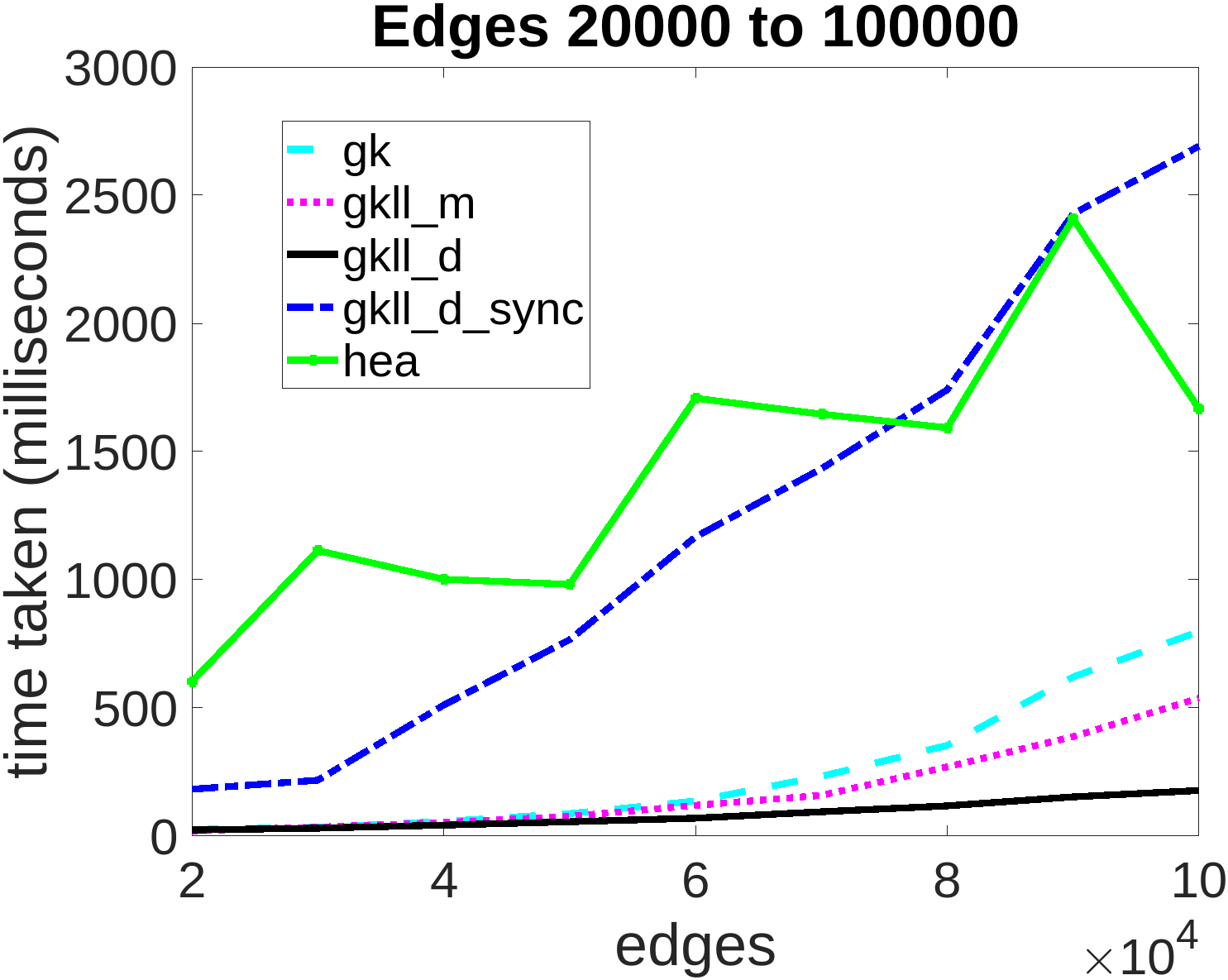}
    }
    \subfigure[]{
        \includegraphics[width=.44\textwidth]{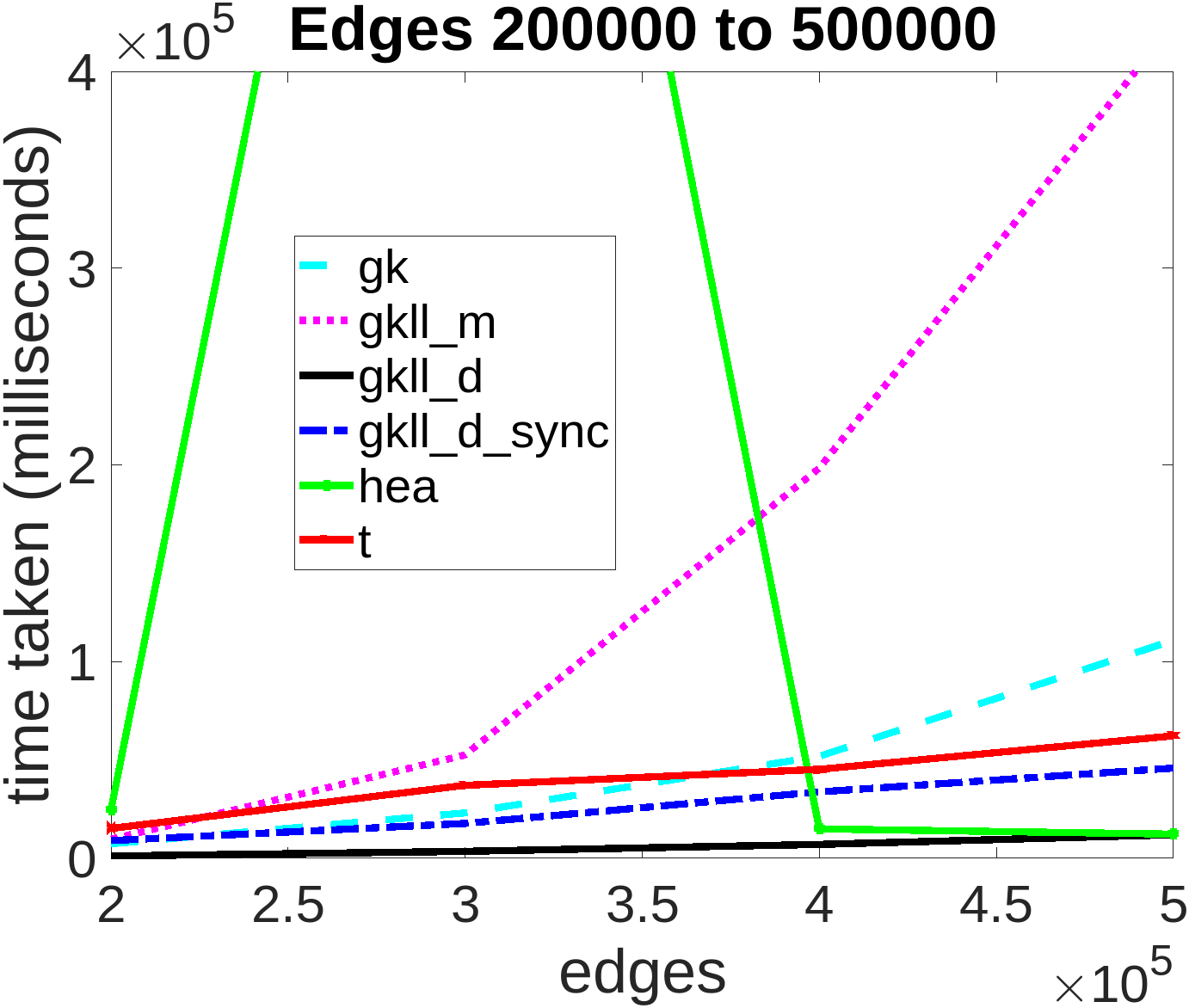}
    }
    \subfigure[]{
        \includegraphics[width=.44\textwidth]{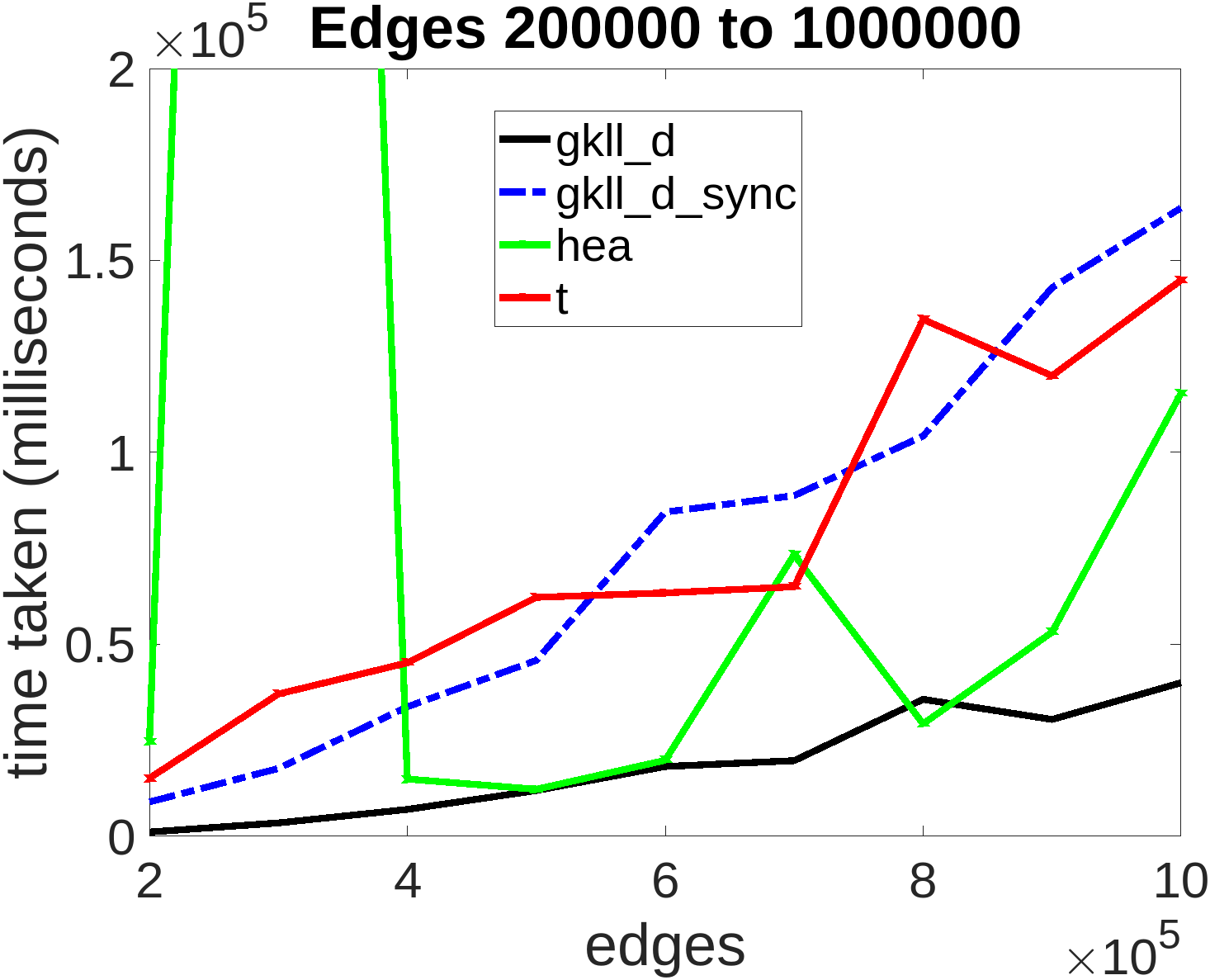}
    }
    \caption{Runtime comparison of \Cref{algorithm:ds-ll-d1}, \Cref{algorithm:ds-ll}, \Cref{algorithm:ds-ellss} and other algorithms for minimal dominating set in the literature. All graphs are of 10,000 nodes.}
    \label{figure:figure}
\end{figure}

We implemented \Cref{algorithm:ds-ll} (\texttt{gkll\_m}), 
\Cref{algorithm:ds-ll-d1} (\texttt{gkll\_d}), lockstep synchronized \Cref{algorithm:ds-ll-d1} (\texttt{gkll\_d\_sync}), \Cref{algorithm:ds-ellss} (\texttt{gk}), algorithms for minimal dominating set present in Hedetniemi et al. (2003) \cite{Hedetniemi2003} (\texttt{hea}) and Turau (2007) \cite{Turau2007} (\texttt{t}), and compare their convergence time. 
The input graphs were random graphs of order 10,000 nodes, generated by the \texttt{networkx} library of python. For comparing the performance results, all algorithms are run on the same set of graphs. 

The experiments are run on Cuda using the \texttt{gcccuda2019b} compiler. \texttt{gkll\_m}, \texttt{gkll\_d} and \texttt{gk} were run asynchronously, and the programs in \texttt{gkll\_d\_sync}, \texttt{hea} and \texttt{t} were run under the required synchronization model.
The experiments are run on \texttt{Intel(R) Xeon(R) Platinum 8260 CPU @} \texttt{2.40} \texttt{GHz, cuda  V100S}. The programs are run using the command \texttt{nvcc $\langle$program$\rangle$.cu}. Here, each multiprocessor ran 256 threads. And, the system provided sufficient multiprocessors so that each node in the graph can have its own thread. All the observations are an average of 3 readings.

\Cref{figure:figure} (a) (respectively, \Cref{figure:figure} (c) and \Cref{figure:figure} (d)) shows a line graph comparision of the convergence time for these algorithms with the number of edges varying from 20,000 to 100,000  (respectively, 200,000 to 500,000 and 200,000 to 1,000,000).
So, the average degree is varying from 4 to 20 (respectively, 40 to 100 and 40 to 200).
\Cref{figure:figure} (b) is same as \Cref{figure:figure} (a), except that the curve for \texttt{hea} is removed so that the other curves can be analyzed closely.
Similarly, \Cref{figure:figure} (c) and \Cref{figure:figure} (d) are similar, however, (1) \Cref{figure:figure} (c) shows curves for convergence time of graphs of average degree 40 to 100, whereas \Cref{figure:figure} (d) shows curves for convergence time of graphs of average degree 40 to 200, and (2) \Cref{figure:figure} (c) contains all 6 curves, whereas \Cref{figure:figure} (d) does not contain curves for \texttt{gkll\_m} and \texttt{gk}.
Observe that
the convergence time taken by the program for \texttt{gkll\_d} is consistently lower than the other algorithms.

In \Cref{figure:figure} (b), it can be observed that the runtime of \texttt{gkll\_m} is lower than the other algorithms (except \texttt{gkll\_d}, which is not surprising). However, in \Cref{figure:figure} (c), it can be observed that the runtime of \texttt{gkll\_m} increases more rapidly and overtakes the runtime of other algorithms (that is why we omitted \texttt{gkll\_m} from \Cref{figure:figure} (d)). This happens mainly because the nodes under \texttt{gkll\_m} are reading values of nodes at distance-4 from themselves. \texttt{gk} also converges comparatively quicker than (\texttt{gkll\_m}) (but not quicker than other algorithms) because its first phase is quicker: the addable nodes move in the dominating set ``carelessly'', whereas in \texttt{gkll\_m} the nodes moving in are ``careful'' as well as the nodes that move out of the dominating set, which adds to the convergence time in the case of \texttt{gkll\_m}.

Next, we discuss how much of the benefit of \texttt{gkll\_d} can be allocated to asynchrony due to the property of lattice-linearity. For this, observe the performance of \texttt{gkll\_d} running in asynchrony (to allow nodes to read old/inconsistent values)  against \texttt{gkll\_d\_sync} (which is the same algorithm as \texttt{gkll\_d} but running in lock-step, to ensure that the nodes only read the most recent values). We observe that the asynchronous implementation has a lower convergence time. This happens mainly because both the asynchronous and the synchronized algorithms have the same convergence time complexity, however, the cost of synchronization (time spent in synchronization, plus the requirement of at least one scheduling thread) is eliminated.

We have performed the experiments on shared memory architecture that allows the nodes to access memory \textit{quickly}. This means that the overhead of synchronization is low. By contrast, if we had implemented these algorithms on a distributed system instead, where computing processors are placed remotely, the cost of synchronization would be even higher. Hence, we anticipate the benefit of lattice-linearity (where synchronization is not needed) to be even higher.

\section{Conclusion}\label{section:conclusion}

In this paper, we 
introduced \textit{fully lattice-linear algorithms} that are tolerant to asynchrony.
Such algorithms induce lattices in the state space even if the underlying problem does not specify, in a suboptimal global state, a set of nodes that must change their states. 



We bridge the gap between lattice-linear problems \cite{Garg2020} and eventually lattice-linear algorithms \cite{Gupta2021}.
%
Fully lattice-linear algorithms overcome the limitations of \cite{Garg2020} and \cite{Gupta2021}. 
Additionally, such algorithms can be developed even for problems that are not lattice-linear. This overcomes a key limitation of \cite{Garg2020} where the system fails if nodes cannot be deemed \imped, or not \imped, naturally.
Since the lattice structures exist in the entire (reachable) state space, we overcome a limitation of \cite{Gupta2021} where only in a subset of global states, lattice-linearity is observed. 

We present algorithms for minimal dominating set (\mds), graph colouring (\gc), minimal vertex cover (\mvc) and maximal independent set (\mis). 
Of these, \mds and \gc relied on tie-breakers on node IDs, a common approach for breaking ties in the literature.  
We observe that a similar design cannot be directly extended to develop an algorithm for  \mvc and \mis. However, the use of complex actions -- that permit a node to change the values of the variables of other node as well as its own -- enable the design of algorithms for \mvc and \mis. We also observe that complex actions can be revised into simple actions -- where a node can only change its own values -- without losing lattice-linearity. However, these revised algorithms utilize the phenomenon of priority inheritance to accomplish this.


In \cite{Garg2020} lattice-linearity is studied in only those systems where the state space forms a distributive lattice where all pairs of global states have a join (supremum) and meet (infimum), and join and meet operations distribute over each other.
We observe that some of these requirements are not required to provide correctness under asynchrony. 
Specifically, we observe that a system allows asynchrony if the state space forms $\prec$-lattices, where the join between any two states is defined, but the definition of meet is not required. This aspect is more overtly observed in instances of \mvc. Specifically, \Cref{figure:vc-semilattice-example} shows that we have a $\prec$-lattice but not a distributive lattice. 

Fully lattice-linear algorithms considered in this paper preserve an advantage of \cite{Garg2020} that was lost in the extension by \cite{Gupta2021}. Specifically, in \cite{Garg2020}, the final configuration could be uniquely determined from the initial state, whereas in \cite{Gupta2021}, all global states (specifically, the infeasible states) do not form a lattice, so starting from an arbitrary state, the state of convergence cannot be predicted. In fully lattice-linear algorithms that we introduce in this paper, the state space is split into multiple lattices and the algorithm starts in one of them. Hence, the state of convergence can be uniquely determined by the initial state.




We have that a lattice-linear algorithm can be transformed to a distance-1 algorithm by having the nodes keep a copy of the variables, of the other nodes, that they want to evaluate their guards with. We transform \Cref{algorithm:ds-ll} to a distance-1 algorithm by using a minimal set of variables needed to evaluate said guards. 

We analyzed the time complexity bounds of an arbitrary algorithm traversing a lattice of states (whether present naturally in the problem or imposed by the algorithm). As corollaries, we obtain the time complexity bounds on the algorithms that we present in this paper. 

We also demonstrate that these algorithms substantially benefit from using the fact that they satisfy the property of lattice-linearity. Specifically, they outperform existing algorithms when they utilize the fact that they are correct without synchronization among processes, i.e., they are correct even if a node is reading old/inconsistent values of its neighbours.




\bibliography{pvll}
\bibliographystyle{acm}

\newpage

\section*{Appendix: Other Abstractions in Concurrent Computing} 

An algorithm is \textit{non-blocking} if in a system running such algorithm, if a node fails or is suspended, then it does not result in failure or suspension of another node. 
A non-blocking algorithm is \textit{lock-free} if system-wide progress can be guaranteed. For example, algorithms for implementing lock-free singly-linked lists and binary search tree are, respectively, presented in \cite{Valois1995} and \cite{Natarajan2014}. In such systems, if a read/write request is blocked then other nodes continue their actions normally.

A non-blocking algorithm is \textit{wait-free} if progress can be guaranteed per node. A wait-free sorting algorithm is studied in \cite{Shavit1997}, which sorts an array of size $N$ using $n\leq N$ computing nodes, and an $O(n)$ time wait-free approximate agreement algorithm is presented in \cite{Attiya1994}. In such systems, in contrast to lock-free systems, it must be guaranteed that all nodes make progress individually.

A key characteristic of lattice linear algorithms is that they  permit the algorithm to execute asynchronously. And, a key difference between non-blocking and asynchronous algorithms is the \textit{system-perspective} for which they are designed.
To understand this, observe that from a perspective, the lattice-linear, asynchronous, algorithms considered in this paper are wait-free.
Each node reads the values of other nodes. Then, it executes an action, if it is enabled, without synchronization. More generally, in an asynchronous algorithm, each node reads the state of its relevant neighbours to check if the guard evaluates to true. It can, then, update its state without coordination with other nodes. 

That said, the goal of asynchronous algorithms is not the progress/blocking of individual nodes 
(e.g., success of insert request in a linked list and a binary search tree, respectively, in \cite{Valois1995} and \cite{Natarajan2014}).
Rather it focuses on the progress from the perspective of the system, i.e., the goal is not about the progress of an action by a node but rather that of the entire system.  
For example, in the algorithm for minimal dominating set present in this paper, if one of the nodes is slow or does not move, the system will not converge. 
However, the nodes can run without any coordination and they can execute on old values, instead of requiring a synchronization primitive to ensure convergence.
In fact, the notion of \imped (recall that in the algorithms that we present in this paper, in any global state, all enabled nodes are \imped) captures this. An \imped node has to make progress in order for the system to make progress.

\textit{Starvation} happens when requests of a higher priority prevent a request of lower priority from entering the critical section indefinitely. To prevent starvation, algorithms are designed such that the priority of pending requests are increased dynamically. Consequently, a low-priority request eventually obtains the highest priority. Such algorithms are called \textit{starvation-free} algorithms. For example, authors of \cite{Kim2005} and \cite{Attiya2010}, respectively, present a starvation-free algorithm to schedule queued traffic in a network switch and a starvation-free distributed directory algorithm for shared objects.
Asynchronous algorithms are starvation-free, as long as all enabled processes can execute. If all enabled processes can execute, convergence is guaranteed.

\textit{Serializability} allows only those executions to be executed concurrently which can be modelled as some permutation of a sequence of those executions. In other words, serializability does not allow nodes to read and execute on old information of each other: only those executions are allowed in concurrency such that reading fresh information, as if the nodes were executing in an interleaving fashion, would give the same result. Serializability is heavily utilized in database systems, and thus, the executions performed in such systems are called \textit{transactions}. Authors of \cite{Papadimitriou1979} show that corresponding to several transactions, determining whether a sequence of read and write operations is serializable is an NP-Complete problem. They also present some polynomial time algorithms that approximate such serializability. Authors of \cite{Fle1982} consider the problem in which the sequence of operations performed by a transaction may be repeated infinitely often. They describe a synchronization algorithm allowing only those schedules that are serializable in the order of commitment.

The asynchronous execution considered in this paper is not \textit{serializable}, especially, since the reads can be from an old global state. Even so, the algorithm converges, and does not suffer from the overhead of synchronization required for serializability.

In \textit{redblue} systems (e.g., \cite{Li2012}), the rules can be divided into two non-empty sets: red rules, which must be synchronized, and blue rules, which can run in a lazy manner and do not have to be synchronized. Lattice-linear and asynchronous systems in general are the systems in which red rules are absent as an enabled node can execute independently without any synchronization.

In \textit{local mutual exclusion}, at a given time, some nodes block other nodes while entering to critical section. This can be done, e.g., by deploying semaphores. Authors of \cite{Keane2001} study the group mutual exclusion problem, where nodes request for various ``sessions'' repeatedly, and it is required that (1) individual processes cannot be in different sessions concurrently, (2) multiple processes can be in the same session concurrently, and (3) is a process tries to enter a session, it is eventually able to do so. Authors of \cite{Khanna2020} propose a leader-based algorithm that deploys local mutual exclusion to solve resource allocation problem in Flying Ad hoc Networks. Authors of \cite{Raymond1989} presented an algorithm for distributed mutual exclusion in computer networks, that uses a spanning tree of the subject network. In this algorithm, the number of messages exchanged per critical section depends on the topology of this tree, typically this value is $O(n)$. Authors of \cite{Yang1995} an algorithm with $O(\lg n)$ time complexity for mutual exclusion among $n$ nodes. Specifically, this algorithm requires atomic reads and writes and in which all spins are local (here a spin means a busy wait in which a node, in this case, waits on locally accessible shared variables).

We see, in algorithms based on local mutual exclusion, that they require additional data structures/variables to ensure that access is provided to (and blocking is deployed on) a certain set of processes. In asynchronous algorithms, nodes do not block each other. In non-lattice-linear problems, we see that usually a tie-breaker is required to ensure the correctness of the executions, however, if a problem is naturally lattice-linear, then this is not required. This is because in the case of non-lattice-linear problems, it may be desired that all unsatisfied nodes do not become enabled, however, in the case of lattice-linear problems, as we see in \cite{Garg2020}, we see that all unsatisfied nodes can be enabled. And, all enabled nodes can read values and perform executions asynchronously, where are allowed to read old values, which is not allowed in algorithms that deploy local mutual exclusion.

\end{document}